\documentclass[english,a4paper]{report}
\usepackage{a4wide}
\pdfoutput=1

\usepackage[T1]{fontenc}
\usepackage[utf8]{inputenc}
\usepackage{babel}

\usepackage{amsmath}
\usepackage{amssymb,mathrsfs} 
\usepackage{amsthm}

\usepackage{framed}
\usepackage{color}
\usepackage{xspace}

\usepackage{enumerate}
\usepackage[normalem]{ulem}
\usepackage{graphicx}
\usepackage{verbatim}

\usepackage[unicode=true,pdfusetitle,
 bookmarks=true,bookmarksnumbered=false,bookmarksopen=false,
 breaklinks=false,pdfborder={0 0 1},backref=false, colorlinks=true,
 linkcolor=black, citecolor=black, filecolor=black, urlcolor=black]{hyperref}

\usepackage[all]{xy}

\definecolor{darkgreen}{rgb}{0,0.6,0}

\makeatletter

\theoremstyle{plain}
\newtheorem{theorem}{Theorem}[chapter]
\newtheorem{lemma}[theorem]{Lemma}
\newtheorem{proposition}[theorem]{Proposition}
\newtheorem*{BT}{Barrington's Theorem}
\theoremstyle{definition}
\newtheorem{definition}[theorem]{Definition}
\newtheorem{corollary}[theorem]{Corollary}
\newtheorem*{corollary*}{Corollary}

\newcommand\complexity@possiblymakesmaller[1]{#1}
\newcommand\complexity@fontcommand{\mathbf}
\newcommand{\ComplexityFont}[1]{%
{\ensuremath{\complexity@possiblymakesmaller{\complexity@fontcommand{#1}}}}
}

\newcommand{\class}[1]{\ComplexityFont{#1}}

\newcommand{\NCone}{\class{NC^1}}

\newcommand{\set}[1]{\left \{#1\right \}}
\newcommand{\Set}[2]{\left\{ #1 :\, #2\right\}}
\renewcommand{\C}{\mathbb{C}}

\newcommand{\N}{\mathbb{N}}

\newcommand{\ket}[1]{| #1 \rangle}
\newcommand{\bra}[1]{\langle #1 |}
\newcommand{\braket}[2]{\langle #1|#2 \rangle}

\newcommand{\id}{\mathbb{I}}

\newcommand{\eps}{\varepsilon}

\newcommand{\PC}{{\sigma}}  
\newcommand{\gh}{\mathit{GH}}    

\newcommand{\PV}{{\sf PV}}
\newcommand{\PVqubit}{\PV_{\text{\rm\tiny \!qubit}}}
\newcommand{\PVmub}{\PV_{\text{\rm\tiny \!MUB}}}

\newcommand{\accept}{\sf ACCEPT}
\newcommand{\reject}{\sf REJECT}

\newenvironment{protocol}{\vspace{-1em}\begin{framed}
}{
\vspace{-3ex}\end{framed}\vspace{-1em}}

\usepackage[Lenny]{fncychap}

\makeatother
\usepackage{lmodern}
\begin{document}

\title{Position-Based Quantum Cryptography and the Garden-Hose Game}
\author{Florian Speelman}
\date{August 2011}


\begin{titlepage}
\begin{center}
{\LARGE Position-Based Quantum Cryptography and the Garden-Hose Game \\}
\vspace{60pt}
\textsc{\large Master's Thesis}\\
\vspace{80pt}
{\LARGE Florian Speelman\\}

\vspace{30pt}
{\Large \emph{Supervisors:}\\
Prof.dr.~Harry Buhrman\\
Dr.~Christian Schaffner\\
}
\vfill
\includegraphics[height=0.65in]{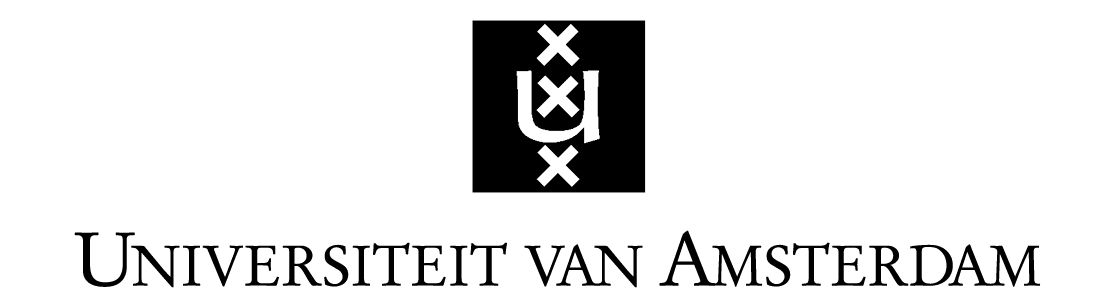}\includegraphics[height=0.95in]{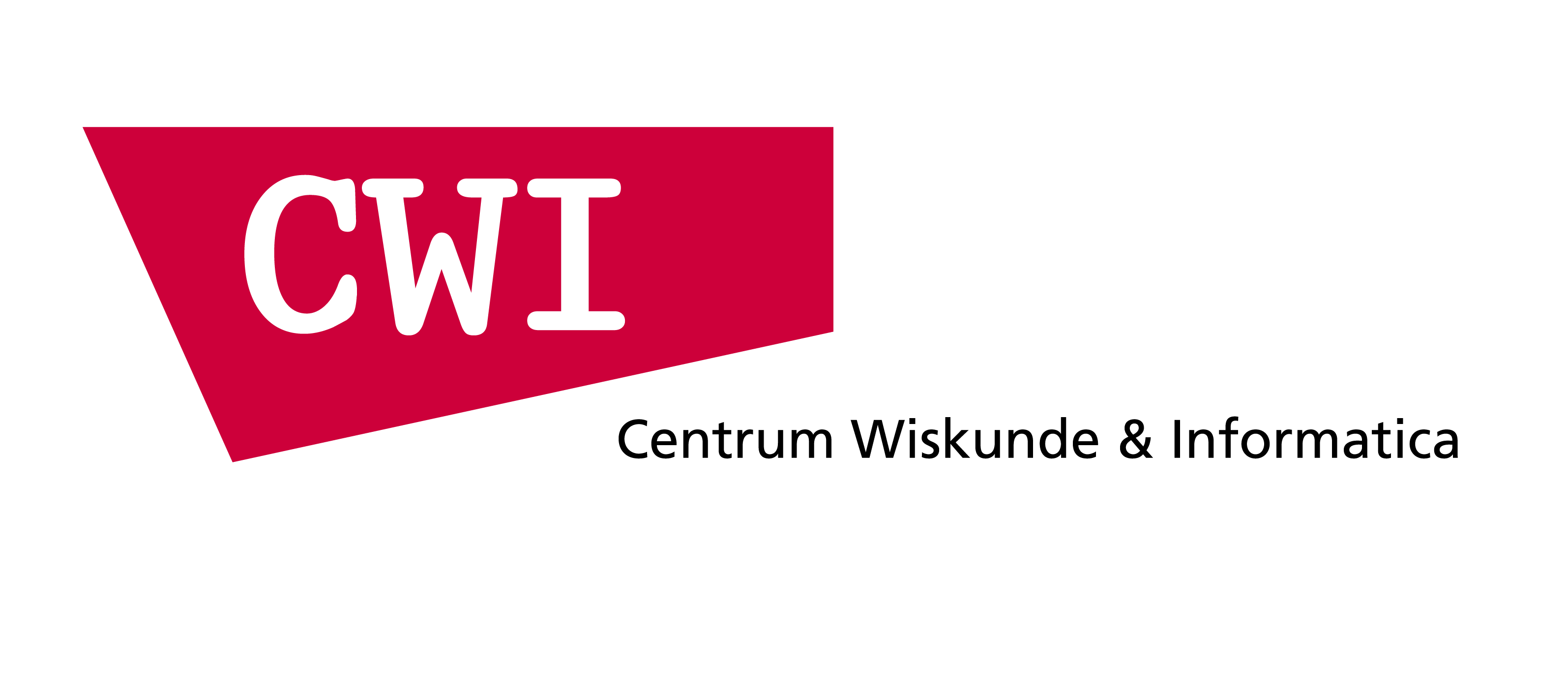}\\
\vspace{2cm}
August 2011
\end{center}

\end{titlepage}

\begin{abstract}
  We study position-based cryptography in the quantum setting. We
  examine a class of protocols that only require the communication of
  a single qubit and $2n$ bits of classical information. To this end,
  we define a new model of communication complexity, the garden-hose
  model, which enables us to prove upper bounds on the number of EPR
  pairs needed to attack such schemes. This model furthermore opens up
  a way to link the security of  position-based quantum cryptography to
  traditional complexity theory.
\end{abstract}

\tableofcontents{}

\chapter{Introduction}


%

\section{Quantum Computing and Quantum Information}
In a now-classic talk, Richard Feynman considered the possibilities of simulating quantum physics
on computers~\cite{Feynman1982}. While it seems to be hard to come up with a way of simulating
general quantum systems, he conjectured that it might be possible to let quantum systems perform
the simulation, a so-called universal quantum simulator. The idea that quantum systems might
be fundamentally better at some computational tasks than computers based on classical physics was a
starting point for the field of \emph{quantum computation}.

The popularity of quantum computation got a big boost after Peter Shor's 1994 discovery of
an efficient algorithm for factoring integers~\cite{Shor1999}.
Besides it being interesting that a quantum algorithm was found that performs much better than the
best known classical algorithm for a well-studied problem, this discovery also has implications for
cryptography.
The security of many cryptographic systems, with RSA being the most widely used,
is based on the assumption that multiplying two large primes is much easier
than factoring the result. If scalable quantum computers would
become available, much of the current encrypted data-traffic could be broken.

Despite the effort going into it, it might still take
a long time for a large-scale quantum computer to be build. Practical quantum computation
might even turn out to be impossible, because of reasons that are not yet known. Still,
currently it seems as though quantum computation is consistent with quantum mechanics, and
a fundamental impossibility result would be a very interesting development.
Applying the tools of quantum information has also led to new results in classical theoretical
computer science, see for example~\cite{DdW11}.

Bennett and Brassard~\cite{Bennett1984} developed a system for quantum key distribution.
The BB84 protocol does not depend on the manipulation of large quantum systems, and implementations
are currently available commercially. The security of the scheme is proven, but recently researchers have demonstrated attacks on implementations of this protocol, using technical properties of the used detectors.
Investigating possible attacks is a currently active area of research on the experimental side of quantum cryptography.

Classical information is made up of bits, while a unit of quantum information is called a \emph{qubit}.
Instead of the classical bit that can take values 0 and 1, a single qubit can be described as a two-dimensional complex unit
vector. We say that a qubit can be in a \emph{superposition} of its basis states $\ket{0}$ and $\ket{1}$, named
in analogy with 0 and 1 of the classical bit. An important characteristic of multiple-qubit states is the ability of
qubits to be \emph{entangled}. A quantum state is entangled when the system of multiple qubits can not be separated into individual qubits without loss of information.
The \emph{EPR pair} or \emph{Bell state} is the maximally entangled state on two qubits,
used in many of the well-known results of quantum information such as teleportation and super-dense coding.

\section{Position-based quantum cryptography}

The goal of \emph{position-based cryptography} is
to perform cryptographic tasks using location as a credential.
The general concept of position-based cryptography
was introduced by Chandran, Goyal, Moriarti and Ostrovsky~\cite{Chandran2009}.
One possible example would be a scheme that encrypts a message 
in such a way that this message can only be read at a certain location, like
a military base. Position authentication is another example of a position-based cryptographic task;
there are many thinkable scenarios in which it would be very useful to
be assured that the sender of a message is indeed at the claimed location.

One of the basic tasks of position-based cryptography is \emph{position
verification}. We have a \emph{prover} $P$ trying to convince
a set of \emph{verifiers} $V_{0},\ldots,V_{k}$, spread around in
space, that $P$ is present at a specific position \emph{pos}. The
first idea for such a protocol is a technique called distance bounding~\cite{BC93}.
Each verifier sends a random string to the prover, using
radio or light signals, and measures how long it takes for the prover
to respond with this string.
Because the signal cannot travel faster than the speed
of light, each verifier can upper bound the distance from the prover.

Before the recent formulation of a general framework for position-based cryptography,
this problem of secure positioning has been studied in the
field of wireless security, and there have been several proposals
for this task (\cite{BC93,SSW,VN04,B04,CH05,SP05,ZLFW06,CCS06}).

Although the security of the proposed protocols can be proven against a single attacker,
they can be broken by \emph{multiple
colluding adversaries}. Multiple adversaries working together can
send a copy of the string sent by the nearest verifier to all other partners
in crime. Each adversary can then emulate the actions of the honest
prover to its closest verifier. It was shown by Chandran et~al.~\cite{Chandran2009}
that such an attack is always possible in the classical world, when not making any extra
assumptions. Their paper does give a scheme where secure position
verification can be achieved, when restricting the adversaries by
assuming there is an upper limit to the amount of information they
can intercept: the Bounded Retrieval Model. Assuming bounded retrieval might
not be realistic in every setting, so the next question was whether other
extensions might be possible to achieve better security. Attention turned
to the idea of using quantum information instead of classical information.
Because the general classical attack depends on the ability of the adversaries to
simultaneously keep information and send it to all other adversaries, researchers
hoped that the impossibility of copying quantum information might make an attack impossible. 
(See Section~\ref{sec:nocloning} for the quantum no-cloning theorem.)

The first schemes for position-based quantum cryptography were investigated
by Kent in 2002 under the name of \emph{quantum tagging}. Together
with Munro, Spiller and Beausoleil, a U.S.~patent was granted for
this protocol in 2006. Their results have appeared in the scientific
literature only in 2010~\cite{Kent2010}. This paper considered several
different schemes, and also showed attacks on these schemes. Afterwards,
multiple other schemes have been proposed, but all eventually turned out to be 
susceptible to attacks.

Finally a general impossibility result was given by Buhrman, Chandran, Fehr, Gelles,
Goyal, Ostrovsky, and Schaffner~\cite{Buhrman2010},
showing that every quantum protocol can be broken. The construction in this
general impossibility result uses a doubly exponential amount of entanglement.
A new construction was recently given by Beigi and K\"onig reducing
this to an exponential amount~\cite{Beigi2011}.

Even though it has been shown that any scheme for position-based quantum cryptography
can be broken, these general attacks use
an amount of entanglement that is too large for use in practical settings.
Even when the honest provers use a small state,
the dishonest players need an astronomical amount of EPR pairs to perform
the attack described in the impossibility proofs. This brings us to the following
question, which is also the central topic of this thesis: How much entanglement
is needed to break specific schemes for quantum position verification?

Answers to this question will take the form of lower bounds and upper bounds.
The upper bounds typically consist of a specific attack on a scheme, showing how it
can be broken using a certain amount of entanglement.
A lower bound is generally harder to prove, since it has to say something about
the amount of resources needed for \emph{any} attack. In their recent paper,
Beigi and K\"onig give a scheme, using Mutually Unbiased Bases, that needs at least a linear number of
EPR pairs to break for honest players~\cite{Beigi2011}. This thesis adds a new upper
bound for this scheme, showing that the lower bound is tight up to a small constant.

In this thesis, we investigate a class of schemes that involves only a single qubit, and $2n$ classical bits.
Such schemes were first considered by Kent et~al.~\cite{Kent2010}.
We focus on the one-dimensional set-up,
but the schemes easily generalize to three-dimensional space. Besides the assumption
that all communication happens at the speed of light, we assume that all parties do
not need time to process the verifiers' messages and can perform computations instantaneously.
We also assume that the verifiers have clocks that are synchronized and accurate, and
that the verifiers have a private channel over which they can coordinate their actions.

The prover wants to
convince the two verifiers, $V_0$ and $V_1$, that he is at position
$pos$ on the line in between them. $V_0$ sends a qubit $\ket{\phi}$
prepared in a random basis to $P$. In addition,
$V_0$ sends a string $x \in \{0,1\}^n$ and $V_1$ a string $y \in \{0,1\}^n$
to $P$. The verifiers $V_0$ and $V_1$ time their actions such 
that the messages arrive at the location of the honest prover at the same time.
After receiving $\ket{\phi}, x$ and $y$, $P$ computes a
predetermined Boolean function $f(x,y)$. He sends $\ket{\phi}$ to
$V_0$ if $f(x,y)=0$ and to $V_1$ otherwise. $V_0$ and $V_1$ check that
they receive the correct qubit in time corresponding to $pos$ and
measure the received qubit in the basis corresponding to which it was
prepared. In order to cheat the scheme, we imagine two provers $P_0$
and $P_1$ on either side of the claimed position $pos$, who try to
simulate the correct behavior of an honest $P$ at $pos$.

Looking from the perspective of the adversaries, we can describe their task in the following way.
$P_0$ receives $\ket{\phi},x$ and $P_1$ receives $y$.
They are allowed to simultaneously
send a single message to each other such that upon receiving
that message they both know $f(x,y)$ and if $f(x,y)=0$ then $P_0$
still has $\ket{\phi}$, otherwise $P_1$ has it in his possession.
The attack described in~\cite{Kent2010} accomplishes this task, for any function~$f$,
but requires an amount of entanglement that is exponential in~$n$.
In this thesis we introduce a complexity measure which relates
to the complexity of computing $f(x,y)$, the garden-hose complexity.
The garden-hose complexity gives an upper bound on the number of EPR pairs the adversaries
need to break the one-qubit scheme that corresponds to the function~$f$.

These protocols are interesting to consider, because the quantum actions of the honest prover
are very simple. All the honest prover has to do is route a qubit to the correct location,
while the verifiers have to measure in the correct basis, actions which are not much harder than
those needed in the BB84 protocol, which is already technologically feasible.
If a gap can be shown between the difficulty of the actions of the honest prover and those of the adversaries,
this protocol would be a good candidate to investigate further for use in real-life settings.
The hope is that for functions $f(x,y)$ that are ``complicated enough'', the amount of entanglement
needed to successfully break the protocol grows at least linearly in the bit length $n$ of the
classical strings $x,y$; we would then require more \emph{classical} computing power of the honest
prover, whereas more \emph{quantum} resources are required by the adversary to break the protocol.
To the best of our knowledge, such a trade-off has never been observed for a quantum-cryptographic protocol.

\section{Complexity Theory}
The field of \emph{computational complexity theory} is concerned with the study of how much
computational resources are required to solve problems.
These resources always have to be defined in relation to a model of computation.

One of the most well-studied models of computation is the \emph{Turing machine},
a hypothetical device that manipulates symbols on a scratch pad according to a set of rules.
To be more precise, the memory of a Turing machine consists of an
\emph{input tape}, one or more \emph{work tapes}, and an \emph{output tape}.
Every tape has a \emph{tape head} that can read and write symbols, one cell at a time. Every time step, the
tape head can move one cell to the left or to the right.
A Turing machine has a finite set of states, and the \emph{state register} of the machine always contains one
of these states. To determine the actions of the Turing machine, the machine has a \emph{transition function}
that, given the current state and the symbols under the tape heads, determines the next state in the state register,
if and what to write under the heads, and the movement of the heads.

The Church-Turing thesis states that if a problem can be solved by an algorithm,
there exists a Turing machine that solves the problem.
When looking at the hardness of a problem, we are concerned with how \emph{efficiently} we can solve it;
we bound resources (such as space used or time used) and then examine which problems can still be solved then.
Many variants of the described (\emph{deterministic})
Turing machine has been studied, such probabilistic Turing machines, non-deterministic Turing machines, and others.
Even though these models can still solve the same problems in principle, they might differ in which problems
they can solve under bounded resources.
The Turing-machine model is surprisingly robust under small modifications; for example,
increasing the number of tapes, or letting the tape head movements only depend on the input
length instead of the actual input, does not change the most commonly used complexity classes.

The first complexity class that we will consider is $\class{P}$.
The class $\class{P}$ contains all problems that can be solved by a deterministic Turing machine
using an amount of time (or steps) that is polynomial in the input length.
This class is often used informally as a notion of \emph{efficient computation}; when a problem is
in $\class{P}$, we will often call it easy (for a conventional computer).

Besides bounding time, it is also possible to bound the amount of tape the Turing machine is allowed to
write on when solving a problem. The complexity class $\class{L}$ contains all problems that can
be solved by a deterministic Turing machine using a logarithmic amount of space.
Complexity theory looks at the relation between complexity classes. For example, we can
say that $\class{L} \subseteq \class{P}$, meaning that all problems that can be solved
in logarithmic space can also be solved in polynomial time. As is often the case
in complexity theory, we are unable to prove that this inclusion is proper, meaning that we can not
prove that the classes are not equal. Still, it is widely believed that there
are problems in $\class{P}$ that are not in $\class{L}$.

An alternative model of computation to the Turing machine is given by the \emph{Boolean circuit}.
A circuit is a directed acyclic graph, where the \emph{input nodes} (which are vertices with
no incoming edges) are given by input bits. The non-source vertices are called \emph{gates},
and represent the logical operations OR, AND, and NOT. Complexity measures
that can be defined on circuits include circuit depth and circuit size. The \emph{depth} of
a circuit is the length of the longest path of an input node to the output node.
The \emph{size} of a circuit is defined as the number of gates of the circuit.

\section{Contributions of This Thesis}

The main theme of this thesis is the analysis of specific protocols
for position-based quantum cryptography.
Large parts of the thesis are based on the article
\emph{The Garden-Hose Game and Application to Position-Based Quantum Cryptography},
by Harry~Buhrman, Serge~Fehr, Christian~Schaffner, and Florian~Speelman, which will be presented at QCRYPT 2011~\cite{BFSS11}.

In Chapter~\ref{sec:attack} we demonstrate that the protocol proposed
by Beigi and K\"onig~\cite{Beigi2011} can be broken using a number of
EPR pairs that is linear in the number of qubits that the honest player
has to manipulate. This improves on their exponential upper bound. The
new upper bound matches their linear lower bound, thereby showing the lower bound
is optimal up to a constant factor.

For the rest of the thesis we turn our attention to protocols for
position verification using only one qubit, complemented with classical
information. In Chapter~\ref{sec:gardenhose} we define the notion
of \emph{garden-hose complexity}, capturing the power of a group of
attacks on these protocols. If a function has garden-hose complexity~$s$,
there is a scheme to break the protocol using $s$ EPR pairs.
We give upper bounds on the garden-hose complexity for several concrete
functions, such as equality, inner product, and majority. Polynomial
upper bounds are also shown for any function that can be computed
in logarithmic space.

Considering the limits of the power of this model, we give an almost-linear
lower bound for the garden-hose complexity of a large group of functions,
including equality and bitwise inner product.
We also show there exist functions having exponential garden-hose
complexity. This exponential lower bound gives evidence that a general efficient
quantum attack, should it exist, has to use other ingredients than just teleportation.

For practical protocols it makes sense to primarily consider functions $f$ which
can be computed in polynomial time. As a consequence of our results,
we find the following interesting connections between the security
of the quantum protocol and classical complexity theory. On one hand,
the assumption that $\class{P}$ is not equal to $\class{L}$
will be needed to keep the possibility open that there exist polynomial-time
functions $f$ where a superpolynomial amount of entanglement is needed
to break our scheme. This also implies a connection in the other direction,
if there is an $f$ in $\class{P}$
such that there is no attack on our scheme using a polynomial number
of EPR pairs, then $\class{P}\neq\class{L}$.

Finally, we give a lower bound for the single-qubit protocol in the
quantum case; it is shown that the adversaries need to manipulate
a number of qubits that is at least logarithmic in the length of the
classical input, while the honest player only has to act on a single
qubit.

These results are steps towards gaining a better understanding of position-based quantum cryptography.
It is realistic to assume that the entanglement between the adversaries is bounded.
The garden-hose model highlights obstacles that a general security proof would have to overcome, and
gives new upper bounds for the amount of entanglement needed for an attack on a specific protocol.
The eventual goal is to prove the security of a specific protocol for position-based quantum cryptography,
under realistic assumptions, so that this new type of cryptography can be considered for implementation
in the real world.

\chapter{Preliminaries\label{sec:preliminaries}}
\section{Basic Quantum Information}
For an introduction into quantum information see~\cite{NC00}.
We will mostly use this section to fix notation.
Quantum states can be described as unit vectors in a complex \emph{Hilbert space}. A Hilbert space
is a complex vector space with an inner product. Throughout this thesis we will typically
use \emph{bra-ket notation}, also known as \emph{Dirac notation}. A quantum mechanical state can be described by a vector $\ket{\psi}$,
called a \emph{ket}. The vector dual to $\ket{\psi}$ is written as $\bra{\psi}$ and is called a \emph{bra}.
The inner product between two states $\ket{\psi}$ and $\ket{\phi}$ is written as $\braket{\psi}{\phi}$.

A single qubit is described in a two-dimensional state space. The most common orthonormal basis we use
for qubits is called the \emph{computational basis} and is defined as
$$\ket{0}=\begin{pmatrix}1 \\ 0 \end{pmatrix},\; \ket{1}=\begin{pmatrix}0 \\ 1 \end{pmatrix} \; .$$

We can describe any one-qubit state by a superposition of these basis vectors, enabling us to write
$$\ket{\psi} = \alpha \ket{0} + \beta \ket{1} \quad \text{with} \quad \alpha,\beta \in \C$$
for any one-qubit state $\ket{\psi}$. Here normalization requires that $|\alpha|^2+|\beta|^2=1$.
In vector notation we would write this state $\ket{\psi}$, and its dual $\bra{\psi}$, as
$$
\ket{\psi}=\begin{pmatrix}\alpha \\ \beta \end{pmatrix},\;
\bra{\psi}=\begin{pmatrix}\alpha^* & \beta^* \end{pmatrix}
$$

The joint state of multiple quantum systems is a vector in a space that is a \emph{tensor product} of
the original spaces. A two-qubit system has computational basis states
\begin{align*}
\ket{0}\ket{0} = \begin{pmatrix} 1 \\ 0 \end{pmatrix} \otimes \begin{pmatrix} 1 \\ 0 \end{pmatrix}  &=
\begin{pmatrix} 1 \\ 0 \\ 0 \\ 0 \end{pmatrix} &
\ket{0}\ket{1} = \begin{pmatrix} 1 \\ 0 \end{pmatrix} \otimes \begin{pmatrix} 0 \\ 1 \end{pmatrix}  &=
\begin{pmatrix} 0 \\ 1 \\ 0 \\ 0 \end{pmatrix} \\ 
\ket{1}\ket{0} = \begin{pmatrix} 0 \\ 1 \end{pmatrix} \otimes \begin{pmatrix} 1 \\ 0 \end{pmatrix}  &=
\begin{pmatrix} 0 \\ 0 \\ 1 \\ 0 \end{pmatrix} &
\ket{1}\ket{1} = \begin{pmatrix} 0 \\ 1 \end{pmatrix} \otimes \begin{pmatrix} 0 \\ 1 \end{pmatrix}  &=
\begin{pmatrix} 0 \\ 0 \\ 0 \\ 1 \end{pmatrix} \; .\\
\end{align*}
The state space of $n$ qubits has dimension $2^n$.
Not all two-qubit states can be written as the tensor product of two one-qubit states. A quantum
state that cannot be written as a product of individual qubit states is said to be \emph{entangled}.
A very important example of an entangled state is the EPR pair, which will be introduced in the next section.

The evolution of a closed quantum system is described by a \emph{unitary transformation}. This means
that we can describe manipulation of the quantum states as a unitary matrix; a matrix
for which holds $U^\dagger U = \id$.

The \emph{Pauli matrices} are four unitary matrices that are
very common in quantum computing. Here we define them as
\begin{align*}
\sigma_0 := \id &:= \begin{pmatrix} 1 & 0 \\ 0 & 1 \end{pmatrix} &
\sigma_1 := X &:= \begin{pmatrix} 0 & 1 \\ 1 & 0 \end{pmatrix} \\
\sigma_2 := Y &:= \begin{pmatrix} 0 & -i \\ i & 0 \end{pmatrix} &
\sigma_3 := Z &:= \begin{pmatrix} 1 & 0 \\ 0 & -1 \end{pmatrix} \; \text{.}
\end{align*}

The \emph{Hadamard matrix} is a unitary transformation which is defined as
\begin{equation*}
H := \frac{1}{\sqrt{2}}\begin{pmatrix} 1 & 1 \\ 1 & -1 \end{pmatrix} \text{.}
\end{equation*}
We write 
\begin{align*}
\ket{+} &:= H\ket{0} = \frac{1}{\sqrt{2}} (\ket{0} + \ket{1}) \quad \text{and}\\
\ket{-} &:= H\ket{1} = \frac{1}{\sqrt{2}} (\ket{0} - \ket{1})
\end{align*}
for the basis vectors of the \emph{Hadamard basis}.

\subsection{Measurement of Quantum States}
A \emph{quantum measurement} is described by a collection $\set{M_m}$ of measurement
operators, where $m$ refers to the measurement outcome. If the state before measurement
is $\ket{\psi}$, the probability that result $m$ occurs is given by
\[
p(m)=\bra{\psi} M^\dagger_m M_m \ket{\psi} \; \text{,}
\]
and the state after getting measurement outcome $m$ is 
\[
\frac{M_m \ket{\psi}}{\sqrt{\bra{\psi} M^\dagger_m M_m \ket{\psi}}} \; .
\]
Reflecting the fact that probabilities sum to one, we have the \emph{completeness relation}
\[
\sum_m{M^\dagger_m M_m} = \id \; .
\]

Measurements in the computational basis can be described by measurement operators
\begin{align*}
M_0 &= \ket{0}\bra{0} & M_1 &= \ket{1}\bra{1} \; ,
\end{align*}
while measurements in the Hadamard basis have measurement operators
\begin{align*}
M_0 &= \ket{+}\bra{+} & M_1 &= \ket{-}\bra{-} \; .
\end{align*}
All measurement operators we use in this thesis are projective measurements.
We call a measurement a \emph{projective measurement} if, beside satisfying the completeness
relation, the $M_m$ are orthogonal projectors. The last requirement means that the operators have to be
Hermitian, that is $M^\dagger_m=M_m$, and that $M_m M_m' = \delta_{m,m'} M_m$. Here $\delta_{m,m'}$ is the
Kronecker delta. As a consequence of using projective measurements,
measuring the same qubit twice, consecutively, will give the same outcome both times.

\subsection{Entanglement and EPR Pairs}
The four different \emph{EPR pairs}, for Einstein, Podolsky and Rosen, or \emph{Bell states} are defined as
\begin{align*}
\ket{\beta_{00}} &:= \frac{1}{\sqrt{2}} \big( \ket{00} + \ket{11} \big) \\
\ket{\beta_{01}} &:= \frac{1}{\sqrt{2}} \big( \ket{01} + \ket{10} \big) \\
\ket{\beta_{10}} &:= \frac{1}{\sqrt{2}} \big( \ket{00} - \ket{11} \big) \\
\ket{\beta_{11}} &:= \frac{1}{\sqrt{2}} \big( \ket{01} - \ket{10} \big) \; \text{.}
\end{align*}
These states will be very often used in this thesis, especially for their use
in \emph{quantum teleportation} (see Section~\ref{sec:teleportation}). Whenever
we will use $\ket{\beta}$, we refer to the state $\ket{\beta_{00}}$.

A measurement on the state $\ket{\beta_{00}}_{AB}$ has the following property:
if qubit $A$ is measured in the computational basis, then
a uniformly random bit $x \in \set{0,1}$ is observed and qubit $B$
collapses to $\ket{x}$. Similarly, if qubit $A$ is measured in the
Hadamard basis, then a uniformly random bit $x \in \set{0,1}$ is
observed and qubit $B$ collapses to $H\ket{x}$.

Using these states, we can define the \emph{Bell measurement}, which projects two qubits into a Bell state,
by measurement operators
\begin{align*}
M_{00} &= \ket{\beta_{00}}\bra{\beta_{00}} & M_{01} &= \ket{\beta_{01}}\bra{\beta_{01}} \\
M_{10} &= \ket{\beta_{10}}\bra{\beta_{10}} & M_{11} &= \ket{\beta_{11}}\bra{\beta_{11}} \; .
\end{align*}

\subsection{The No-Cloning Theorem}
\label{sec:nocloning}
The \emph{no-cloning theorem} is a classic result of quantum information
which states that it is impossible to copy an arbitrary quantum state.
This theorem has very important consequences for quantum cryptography.
Without the impossibility of cloning the BB84 scheme would
be insecure, for example. The no-cloning theorem is also
the reason why the classical attack on schemes for position-based cryptography
does not generalize to the quantum case.

\begin{theorem}
There exists no unitary operation $U$ that perfectly copies the state of an arbitrary qubit.
\end{theorem}
\begin{proof}
By contradiction, suppose we have a unitary operation $U$ that performs a copy, so that
$U\ket{\psi}\ket{s}=U\ket{\psi}\ket{\psi}$ for any possible $\ket{\psi}$, where $\ket{s}$ is some starting
state that is independent of $\ket{\psi}$.
More specifically, this would imply
\[
U(\ket{0}\ket{s}) = \ket{0}\ket{0}
\]
and also
\[
U(\ket{1}\ket{s}) = \ket{1}\ket{1} \text{.}
\]
But now let us try to copy $\ket{+}=\frac{1}{\sqrt{2}}(\ket{0}+\ket{1})$.
Since $U$ is linear, we can use the previous equations to get
\[
U\ket{+}\ket{s} = \frac{1}{\sqrt{2}}(U\ket{0}\ket{s}+U\ket{1}\ket{s})=\frac{1}{\sqrt{2}}(\ket{0}\ket{0}+\ket{1}\ket{1})
\]
and this is not equal to $\ket{+}\ket{+}$, giving a contradiction.
\end{proof}

\subsection{Quantum Teleportation}\label{sec:teleportation}
The goal of quantum teleportation is to transfer a quantum state from
one location to another by only communicating classical information.
Teleportation requires pre-shared entanglement among the two
locations.

Let us say Alice wants to teleport a qubit~$Q$ to Bob,
in an arbitrary unknown state $$\ket{\psi}_Q = \alpha \ket{0}_Q + \beta \ket{1}_Q \; .$$
Alice and Bob share a quantum state $\ket{\beta_{00}}_{AB}$, where Alice has qubit~$A$ and Bob
has~$B$. Define the total state of their system
$\ket{\Psi} = \ket{\beta_{00}}_{AB} \ket{\psi}_Q$.

We can rewrite the state of their quantum system as
\begin{align*}
\ket{\Psi} =&\; \frac{1}{\sqrt{2}} \big( \ket{0}_A\ket{0}_B + \ket{1}_A\ket{1}_B \big) \big( \alpha \ket{0}_Q + \beta \ket{1}_Q \big)\\
=&\; \frac{1}{2} \Big[ \ket{\beta_{00}}_{AQ} \big( \alpha \ket{0}_B + \beta \ket{1}_B \big) +
	\ket{\beta_{01}}_{AQ} \big( \alpha \ket{1}_B + \beta \ket{0}_B \big) 
		\\
&\quad+	\ket{\beta_{10}}_{AQ} \big( \alpha \ket{0}_B - \beta \ket{1}_B \big)
 +\ket{\beta_{11}}_{AQ} \big( \alpha \ket{1}_B - \beta \ket{0}_B \big)
		\Big] \;\text{.}\\
\end{align*}
Note that if we fill in the $\ket{\beta_{z}}$ terms, we get exactly the same state as we started with;
all that happened so far is a re-ordering of terms. Now Alice performs a Bell measurement on qubits $A$ and $Q$,
getting an outcome $z \in \set{00, 01, 10, 11}$. After this measurement, the state Bob holds
will be equal to $\PC_k \ket{\psi}$, where $\PC_k$ is a Pauli correction depending on the outcome $z$.
Now Alice sends the two bits $z$ to Bob.

We can quickly check that when $z=00$, Bob does not have to apply a correction. On $z=01$, Bob can recover $\ket{\psi}$ by
applying $\PC_1=X$. When $z=10$, Bob has to apply $\PC_3=Z$. And when $z=11$, Bob can recover the original
state $\ket{\psi}$ by applying $\PC_2=Y$ to his qubit. The $Y$ operation does contain an extra factor $i$ in its
usual definition, but this only adds a global phase to the quantum state, which we can always ignore.
With this protocol, Alice can effectively transfer
a quantum state to Bob, using a pre-shared entangled state and classical information.

\section{Barrington's Theorem}\label{sec:barringtonstheorem}
A classic result in computational complexity, Barrington's theorem~\cite{Barrington1989}
was a resolution of a long-standing open problem concerning the power of a model called 
\emph{bounded-width computations}. Many theorists believed the model to be not very powerful,
and hoped to prove lower bounds in the model, until Barrington showed it was able to
simulate a large class of circuits. In this thesis, we apply the construction Barrington
used in his original proof. This construction enables us to find strategies
for functions computed by this class of circuits in the garden-hose model (Section \ref{sec:logdepthcircuits}).

A \emph{cycle} is a permutation which maps some subset of elements to each other in a cyclic way.
When we explicitly write out a permutation, we will use \emph{cycle notation}.
In cycle notation, the permutation $\mu=(a_1 \dots a_k)$ has the action
$$a_1 \mapsto a_2 \mapsto \dots \mapsto a_k \mapsto a_1 \text{.}$$
So for each index $i$, we have $\mu(a_i)=a_{i+1}$, where $a_{k+1}$ refers
back to $a_1$. Any other permutation can be written as a product of disjoint cycles.

In the context of this section, let an \emph{instruction over $S_5$} be a triple $(i,\mu,\nu)$,
where $i$ is the index to a bit of the input, and $\mu$ and $\nu$ are permutations
in the symmetric group $S_5$. Let $e$ be the identity permutation and 
$x_1, x_2, \dots, x_n$ be a list of the $n$ inputs to the circuit.
The instruction $(i,\mu,\nu)$ evaluates to $\mu$ if $x_i$ is true and to $\nu$ if
$x_i$ is false. A \emph{width-5 permutation branching program (5-PBP)} is a sequence of instructions over $S_5$,
and it evaluates to the product of the value of its instructions. The length of a program is
the number of instructions. We will say that a permutation program $P$ \emph{computes a circuit
$C$ with output $\mu$} if it evaluates to a cycle $\mu$ whenever $C$ is true, and to the identity $e$ if $C$ is false.

\begin{BT}
Given a boolean circuit $C$ of fan-in two and
depth $d$, with $\mu$ a five-cycle in $S_{5}$, there is a 5-PBP of length at most $4^{d}$ 
that evaluates to $\mu$ if $C$ evaluates to true and to the identity if $C$ evaluates to false.
\end{BT}

\begin{lemma}\label{lemma:bt1}
A program is independent of its output $\mu$. If a 5-PBP
evaluates to $\mu$ if $C$ is true and to $e$ if $C$ is false,
there exists a 5-BPB of equal length that evaluates to a given five-cycle $\nu$ instead of $\mu$.
\end{lemma}
\begin{proof}
There exists a permutation $\theta$ such that $\nu=\theta \mu \theta^{-1}$. Multiply both
permutations in the first instruction by $\theta$ on the left, and multiply both
parts of the last instruction by $\theta^{-1}$ on the right. The new program
has the same length as the old and produces $\nu$ if $C$ is true and $e$ is $C$ is false.
\end{proof}

\begin{lemma}\label{lemma:bt2}
If a permutation branching program $P$ computes $C$, there is also a program $P'$ of equal length that
computes the negation of $C$.
\end{lemma}
\begin{proof}
Let $P$ compute $C$ with output $\mu$, and let the last instruction be $(i,\tau,\upsilon)$.
We make $P'$ identical to $P$ except for the last instruction, $(i,\tau \mu^{-1}, \upsilon \mu^{-1})$.
This new program evaluates to $e$ if $C$ is true and to $\mu^{-1}$ if $C$ is false.
By Lemma~\ref{lemma:bt1} we can make it evaluate to any five-cycle.
\end{proof}

\begin{lemma}\label{lemma:bt3}
There are two five-cycles $\mu_1$ and $\mu_2$ in $S_5$ whose commutator is a five-cycle unequal
to the identity. The commutator of two permutations $a$ and $b$ is defined as $aba^{-1}b^{-1}$.
\end{lemma}
\begin{proof}$(12345)(13542)(54321)(24531)=(13254)$.
\end{proof}

\begin{proof}[Proof of Barrington's theorem]
The proof closely follows Barrington's original proof~\cite{Barrington1989}.
Let $C$ be a circuit of depth $d$, fan-in two,
consisting of AND and OR gates.
We prove the statement by induction on $d$.
In the base case, if $d=0$, the circuit has no gates, so it is easy to use one instruction that evaluates
correctly to an input (or negation of an input).

Now assume without loss of generality that for depth greater than zero
the output gate of $C$ is an AND gate. If the
output gate is an OR gate we can use Lemma~\ref{lemma:bt2} to turn it into an AND gate
without gaining any length. The inputs to this gate, $C_1$ and $C_2$, are circuits of depth $d-1$.
By the induction hypothesis these circuits can be computed by 5-PBPs of length at most $4^{d-1}$,
call these $P_1$ and $P_2$.
Now let $P_1$ compute $C_1$ with output $\mu_1$ and $P_2$ compute $C_2$ with output $\mu_2$, with $\mu_1$ and $\mu_2$ as
defined in Lemma~\ref{lemma:bt3}. By Lemma~\ref{lemma:bt1} we can choose the
output of the programs without increasing the length.
Let $P'_1$ be equal to $P_1$ but instead with output $\mu^{-1}_1$.
Similarly, make $P'_2$ equal to $P_2$ but with output $\mu^{-1}_2$.
Let the program $P$ be the concatenation $P_1 P_2 P'_1 P'_2$. If either $C_1$ or $C_2$ is false,
$P$ evaluates to the identity permutation. If both are true, $P$ evaluates to the commutator
of $\mu_1$ and $\mu_2$, which is a five-cycle.
(By the first lemma, we can replace this five-cycle by any other.)
$P$ has length at most $4^d$.
\end{proof}

\begin{corollary*}
Every problem in $\NCone$ can be decided by a 5-width permutation branching program of polynomial length.
\end{corollary*}
\begin{proof}
If a problem is in $\NCone$, there is a circuit $C$ with logarithmic depth $d$,
fan-in two, consisting of AND, OR, and NOT gates that decides the problem, by the definition of $\NCone$.
Having $d=O(\log{n})$, we can use Barrington's theorem to 5-BPB of length at most $4^{d}=4^{O(\log{n})}$
that computes $C$.\footnote{All logarithms in this thesis are with respect to base 2.} Since $4^{O(\log{n})}$ is polynomial in $n$, the statement directly follows.
\end{proof}

\chapter{Attacking a protocol for position verification}\label{sec:attack}

\section{Mutually Unbiased Bases}
We use the following standard definition of mutually unbiased bases. The constructions
and notations that we base our attack on were introduced in~\cite{Lawrence2002}.
For a construction that also works for states of dimension other than powers of two,
see~\cite{Bandyopadhyay2002}.

\begin{definition}
Two orthonormal bases $\{\ket{e^a_i}\}_{i=1,\ldots,d}$ and
$\{\ket{e^b_j}\}_{j=1,\ldots,d}$
of
$\mathbb{C}^d$ are called \emph{mutually
  unbiased}, if $|\braket{e^a_i}{e^b_j}|^2 = \frac{1}{d}$ holds for all $i,j \in
\set{1,\ldots,d}$.
\end{definition}


A \emph{Pauli operator} on an $n$-qubit state is the tensor product of
$n$ one-qubit Pauli matrices. Hence, there are $4^n$ Pauli operators in
total. For $i \in \set{0,1,2,3}^n$, we can write the Pauli operator
$O_i$ as 
\[
O_i = \sigma^1_{i_1} \sigma^2_{i_2} \dots \sigma^n_{i_n} = \prod_{k=1}^{n}\sigma_{i_k}^{k}
\]
where $\sigma^k_j$ is the $j$-th Pauli matrix acting on qubit $k$
(tensored with the identity on the other qubits).


Excluding the identity, there are $4^n-1$ Pauli operators. These can
be partitioned in $2^n + 1$ distinct subsets consisting of $2^n - 1$
\emph{commuting} operators each~\cite{Lawrence2002}. The $2^n$ common
eigenvectors of such a set of $2^n-1$ commuting operators define an
orthonormal basis. It can be shown that for any such partitioning, the
resulting $2^n+1$ bases are pairwise mutually
unbiased~\cite{Lawrence2002}. We denote by $\ket{e^a_x}$ the $x$-th
basis vector of the $a$-th mutually unbiased basis of this
construction, where $x \in \set{0,1}^n$ and $a \in \mathcal{A}$ for a
set $\mathcal{A}$ of $2^n+1$ elements.

In the following, we will exploit a special property of this
construction of mutually unbiased bases in order to attack a protocol
for position-verification recently proposed by Beigi and
K\"onig~\cite{Beigi2011}. In particular, we use the fact that applying
a Pauli operator only permutes the basis vectors \emph{within} every
mutually unbiased basis, but does not map any basis vector into
another basis. This property is captured by the following lemma.

\begin{lemma} \label{lem:mubtrick} Let $U$ be an arbitrary Pauli
  operator on $n$ qubits. For arbitrary $a \in \mathcal{A}$ and $x \in
  \set{0,1}^n$, let $\ket{e^a_x}$ be the $x$-th basis vector of the
  $a$-th mutually unbiased basis obtained from the construction
  above. Then, there exists $z \in \set{0,1}^n$ such that $U
  \ket{e^a_x} = \ket{e^a_z}$.
\end{lemma}
\begin{proof}
We can write $U$ as
\[
U=\sigma^1_{r_1} \sigma^2_{r_2} \dots \sigma^n_{r_n} =
\prod_{k=1}^{n}\sigma_{r_k}^{k} \, .
\]
Assume $\ket{e^a_x}$ is a common eigenvector of an internally
commuting subset $A$ of the Pauli operators, like described
earlier. Denote the $2^n-1$ elements of $A$ by $O^A_\ell$ with $\ell \in
\{1,\ldots,2^n-1 \}$.  Note that $\sigma_0 \sigma_i = \sigma_i
\sigma_0$ for $i \in \set{0,1,2,3}$ and $\sigma_i \sigma_j = (-1)^{\delta_{ij}}
\sigma_j \sigma_i$ for $i,j \in \set{1,2,3}$ and $\delta_{ij}$ the
Kronecker $\delta$-function.  Because $\ket{e^a_x}$ is a common
eigenvector of the Pauli operators in this set, it holds for all $\ell$ that $O^A_\ell
\ket{e^a_x} = \lambda_\ell \ket{e^a_x}$ for some eigenvalue $\lambda_\ell$.
To prove the claim, we show that $U \ket{e^a_x}$ is also an
eigenvector of all $O^A_\ell$, with some (possibly different) eigenvalue $\lambda'_\ell$.
\begin{align*}
O^A_\ell U \ket{e^a_x} &= \prod_{k=1}^{n} \sigma_{\ell_k}^{k} \sigma_{r_k}^{k} \ket{e^a_x} \\
&= (-1)^{\alpha(r,\ell)} \prod_{k=1}^{n} \sigma_{r_k}^{k} \sigma_{\ell_k}^{k} \ket{e^a_x} \\
&= (-1)^{\alpha(r,\ell)} U O^A_\ell \ket{e^a_x}\\
&= \lambda'_\ell U \ket{e^a_x} \, ,
\end{align*}
where we define $\lambda'_\ell := (-1)^{\alpha(r,\ell)} \lambda_\ell$
and the function $\alpha(r,\ell)$ determines the phase arising from
the commutation relations of the $\sigma_{r_k}\!$'s and
$\sigma_{\ell_k}\!$'s. Because $U \ket{e^a_x}$ is a common
eigenvector of all $O^A_\ell$, there exists $z \in \set{0,1}^n$ such
that $\ket{e^a_z} = U \ket{e^a_x}$.
\end{proof}

\section{The Protocol}
The protocol described in Figure~\ref{fig:PVmub} uses an (almost)
complete set of mutually unbiased bases
$\set{\ket{e^a_x}_{x=1,\ldots,2^n}}^{a \in \set{0,1}^n}$ as defined
above. The protocol can be seen as a higher-dimensional extension of
the basic BB84-protocols proposed and analyzed
in~\cite{Kent2010,Buhrman2010}. In~\cite{Beigi2011}, Beigi and K\"onig
show that $\PVmub$ is secure against adversaries that share fewer than
$n/2$ EPR pairs and are restricted to one round of
simultaneous \emph{classical} communication. They leave open whether
the protocol remains secure against colluding adversaries that share
more entanglement. We answer this question here. In the rest of the
section, we show that for the construction of MUBs mentioned above, it
is sufficient for adversaries to share $n$ EPR pairs in order to
perfectly break the protocol $\PVmub$. It follows that the lower bound
on EPR pairs given in~\cite{Beigi2011} is optimal up to constant factors.

\begin{figure}[htb]
\begin{protocol} 
\begin{enumerate}\setlength{\parskip}{0.1ex}\setcounter{enumi}{-1}
\item[0.]  $V_0$ and $V_1$ share common (secret) randomness in the
  form of uniformly distributed bit strings $a,x\in
  \set{0,1}^n$. \label{proto:positionbased}
\item[1.] $V_0$ sends~$a$ to~$P$ and $V_1$ prepares the state
  $\ket{e^a_x}$ and sends it to~$P$. The timing is chosen such that
  both the classical information and the quantum state arrive at the
  prover at the same time.
\item[2.] $P$ measures the state in the basis $\{\ket{e^a_i}\}_i$,
  getting measurement outcome $\hat{x}\in\set{0,1}^n$. He sends $\hat{x}$
  to both $V_0$ and~$V_1$.
\item[3.] $V_0$ and $V_1$ accept if they receive $\hat{x}$ at times
  consistent with~$\hat{x}$ being emitted from the claimed position in both
  directions simultaneously, and $\hat{x}=x$. 
\end{enumerate}
\end{protocol}
 \caption{Protocol $\PVmub$ from~\cite{Beigi2011} for position-verification using mutually
   unbiased bases. \label{fig:PVmub}}
\end{figure}

\section{The Attack}

The attack reported here is very similar to the attack on the
BB84-scheme described in~\cite{Kent2010}. The colluding adversaries
$\tilde{P}_0$ en $\tilde{P}_1$ set up between the prover's claimed
position and the verifiers $V_0$ and $V_1$, intercepting messages from
$V_0$ and $V_1$.


Adversary $\tilde{P}_0$ has knowledge of the basis $a$ and
$\tilde{P}_1$ gets the state $\ket{e^a_x}$.  Our attack shows that
using $n$ EPR pairs and one round of simultaneous classical communication
suffices to determine $x$, and thus breaking protocol $\PVmub$. We
assume that the set of mutually unbiased bases used is equivalent to a
basis obtained by a partitioning of Pauli operators as described
above. To the best of our knowledge, any currently known construction
of mutually unbiased basis sets of dimension $2^n$ is of this form.
If the used set of mutually unbiased bases differs from one of these
by a unitary transform, the attack still works by the adversaries just applying
this unitary before the first step.

As soon as $\tilde{P}_1$ receives the state $\ket{e^a_x}$, she teleports
it to $\tilde{P}_0$ and forwards the classical outcome of the
teleportation measurement indicating the needed Pauli
correction. Using Lemma~\ref{lem:mubtrick}, the teleported state is
still a basis vector of the same mutually unbiased basis, i.e.~the
state $\tilde{P}_0$ has before correction is $\ket{e^a_z}$, with $z$
depending on the teleportation measurement outcome.  $\tilde{P}_0$
measures $\ket{e^a_z}$ in basis $a$, getting outcome $z$ which she
sends to $\tilde{P}_1$.

Now both adversaries know $a$, $z$ and the teleportation correction,
which is all the necessary information to obtain $x$.
In principle, they can now reconstruct $\ket{e^a_z}$, apply the Pauli correction getting $\ket{e^a_x}$ and
measure in basis $a$. In practice, they can also find $x$ by classically computing which $x$
corresponds to which correction, $a$, and $z$, instead of needing to reconstruct the entire state.
%

\chapter{The Garden-Hose Game} \label{sec:gardenhose}
\section{Motivation} \label{sec:motivation}
The results of this section are motivated by the study of a particular
quantum protocol for secure position verification, described in
Figure~\ref{fig:PVqubit}. The protocol is of the generic form
described in Section~3.2 of~\cite{Buhrman2010}. In
Step~\ref{step:preparation}, the verifiers prepare challenges for the
prover. In Step~\ref{step:send}, they send the challenges, timed in
such a way that they all arrive at the same time at the prover. In
Step~\ref{step:prover}, the prover computes his answers and sends them
back to the verifiers. Finally, in Step~\ref{step:verification}, the
verifiers verify the timing and correctness of the answer.

As in~\cite{Buhrman2010}, we consider here for simplicity the case
where all players live in one dimension, the basic ideas generalize to
higher dimensions. In one dimension, we can focus on the case of two
verifiers $V_0, V_1$ and an honest prover $P$ in between them.

We minimize the amount of quantum communication in that only one
verifier, say $V_0$, sends a qubit to the prover, whereas both
verifiers send classical $n$-bit strings $x,y \in \set{0,1}^n$ that
arrive at the same time at the prover. We fix a publicly known boolean
function $f: \set{0,1}^n \times \set{0,1}^n \rightarrow \set{0,1}$
whose output $f(x,y)$ decides whether the prover has to return the
qubit (unchanged) to verifier $V_0$ (in case $f(x,y)=0$) or to
verifier $V_1$ (if $f(x,y)=1$).

\begin{figure}[htb]
\small
\begin{protocol}
\begin{enumerate}\setlength{\parskip}{0.1ex}\setcounter{enumi}{-1}
\item\label{step:preparation} $V_0$ randomly chooses two $n$-bit
  strings $x,y \in \set{0,1}^n$ and privately sends $y$
  to~$V_1$. $V_0$ prepares an EPR pair $(\ket{0}_V\ket{0}_P +
  \ket{1}_V\ket{1}_P)/\sqrt{2}$. If $f(x,y)=0$, $V_0$ keeps the qubit
  in register $V$. Otherwise, $V_0$ sends the qubit in register $V$
  privately to $V_1$.
\item\label{step:send} $V_0$ sends the qubit in register $P$ to the prover $P$ together
  with the classical $n$-bit string $x$. $V_1$ sends $y$ so that it
  arrives at the same time as the information from $V_0$ at $P$.
\item\label{step:prover} $P$ evaluates $f(x,y) \in \set{0,1}$ and
  routes the qubit to $V_{f(x,y)}$.
\item\label{step:verification} $V_0$ and $V_1$ accept if the qubit
  arrives in time at the right verifier and the Bell measurement of
  the received qubit together with the qubit in $V$ yields the correct
  outcome.
\end{enumerate}
\end{protocol}
 \caption{\!\mbox{Position-verification scheme $\PVqubit$ using one
     qubit and classical $n$-bit strings.}}
 \label{fig:PVqubit}
\end{figure}

The motivation for considering this protocol is the following: As the
protocol uses only one qubit which needs to be correctly routed, the
honest prover's quantum actions are trivial to perform. His main task
is evaluating a classical boolean function~$f$ on classical inputs $x$
and $y$ whose bit size $n$ can be easily scaled up. On the other hand,
our results in this section suggest that the adversary's job of
successfully attacking the protocol becomes harder and harder for
larger input strings $x,y$. The hope is that for ``complicated
enough'' functions $f(x,y)$, the amount of EPR pairs needed to
successfully break the security of the protocol $\PVqubit$ grows (at
least) linearly in the bit length $n$ of the classical strings $x,y$.

If this intuition can be proven to be true, it is a very
interesting property of the protocol that we obtain a favorable
relation between quantum and classical difficulty of operations in the
following sense: if we increase the length of the classical inputs
$x,y$, we require more \emph{classical} computing power of the honest
prover, whereas more \emph{quantum} resources (EPR pairs) are required by
the adversary to break the protocol. To the best of our knowledge,
such a trade-off has never been observed for a quantum-cryptographic
protocol. 

In order to analyze the security of the protocol $\PVqubit$, we define
the following communication game in which Alice and Bob play the roles
of the adversarial attackers of $\PVqubit$. Alice starts with an
unknown qubit $\ket{\phi}$ and a classical $n$-bit string $x$
while Bob holds the $n$-bit string $y$. They also share some quantum
state $\ket{\eta}_{AB}$ and both players know the Boolean function $f:
\set{0,1}^n \times \set{0,1}^n \to\{0,1\}$. The players are allowed
one round of simultaneous classical communication combined with
arbitrary local quantum operations. When $f(x,y)=0$, Alice should be in
possession of the state $\ket{\phi}$ at the end of the protocol and on
$f(x,y)=1$, Bob should hold it.

As a simple example consider the case where $f(x,y)=x\oplus y$, the
exclusive OR function, with 1-bit inputs $x$ and $y$. Alice and Bob
then have the following way of performing this task perfectly by using
a pre-shared quantum state consisting of three EPR pairs. Label the first two EPR pairs $0$
and $1$. Alice teleports $\ket{\phi}$ to Bob using the pair labeled
with her input $x$. This yields measurement result $i\in\{0,1,2,3\}$,
while Bob teleports his half of the EPR pair labeled $y$ to Alice
using his half of the third EPR pair while obtaining measurement
outcome $j \in \set{0,1,2,3}$ . In the round of simultaneous
communication, both players send the classical measurement results and
their inputs $x$ or $y$ to the other player. If $x\oplus y=1$,
i.e.~$x$ and $y$ are different bits, Bob can apply the Pauli operator
$\sigma_{i}$ to his half of the EPR pair labeled $x=y\oplus 1$,
correctly recovering $\ket{\phi}$. Similarly, if $x\oplus y=0$, it is
easy to check that Alice can recover the qubit by applying
$\sigma_{i}\sigma_{j}$ to her half of the third EPR pair.

If Alice and Bob are {\em constrained} to the types of actions in the example above, i.e., if they are restricted to teleporting the quantum state back and forth depending
on their classical inputs, this leads to the following notion of garden-hose game and garden-hose complexity. 

\section{Definition of the Garden-Hose Game}

Alice and Bob get $n$-bit input strings $x$ and $y$, respectively.
Their goal is to ``compute'' an agreed-upon Boolean function $f:
\set{0,1}^n \times \set{0,1}^n \to \set{0,1}$ on these inputs, in the
following way.  We assume that Alice and Bob have $s$ pipes between
them. Depending on their respective classical inputs $x$ and $y$, they
connect their ends of the pipes with pieces of hose, of which they
have an unlimited amount.  Note however, that we do not allow
``T-pieces'' (or more complicated constructions) of hose which connect
two or more pipes to one, or vice versa; only one-to-one connections
are allowed.  Alice has a source of water which she connects to one of
the pipes, and then she turns on the water. It is easy to check that
the water will flow out on either of the sides, i.e.~no ``deadlocks'' are possible.
The players succeed in computing
$f$ (we may also say: they win the garden-hose game), if the water
comes out of one of the pipes on Alice's side whenever $f(x,y) = 0$,
and the water comes out of one of the pipes on Bob's side whenever
$f(x,y) = 1$.  Note that it does not matter out of which pipe the
water flows, only on which side it flows.  We stress once more that
what makes the game non-trivial is that Alice and Bob must do their
``plumbing'' based on their local input only, and they are not allowed
to communicate. We refer to Figure~\ref{fig:xor} for an illustration
of computing the XOR function in the garden-hose model.

\begin{figure}
\begin{framed}
\center
\includegraphics{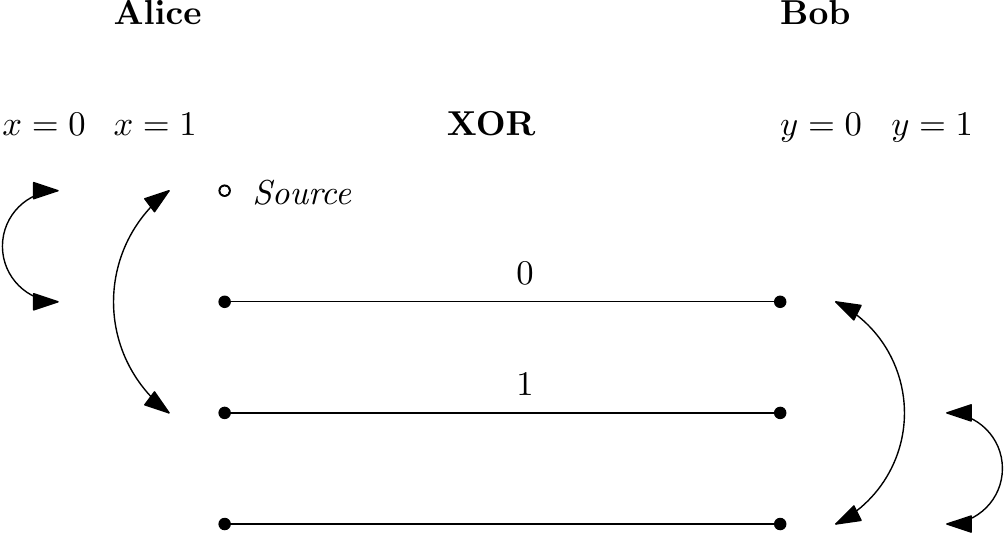}
\end{framed}
\caption{Garden-hose game for the XOR function. \label{fig:xor}}
\end{figure}

We can translate any strategy of Alice and Bob in the garden-hose game
to a perfect quantum attack of $\PVqubit$ by using one EPR pair per
pipe and performing Bell measurements where the players connect the
pipes.  Our hope is that also the converse is true in spirit: if many
pipes are required to compute $f$, say we need superpolynomially many,
then the number of EPR pairs needed for Alice and Bob to successfully
break $\PVqubit$ with probability close to $1$ by means of an {\em
  arbitrary} attack (not restricted to Bell measurements on EPR pairs)
should also be superpolynomial. We leave this as an interesting
problem for future research. 
We stress that for this application, a polynomial lower bound on the number of pipes,
and thus on the number of EPR pairs, is already interesting.

We formalize the above description of the garden-hose game, given
in terms of pipes and hoses etc., by means of rigorous graph-theoretic
terminology.  However, we feel that the above terminology captures the
notion of a garden-hose game very well, and thus we sometimes use the
above ``watery'' terminology.  We start with a balanced bi-partite
graph $(A\cup B,E)$ which is 1-regular and where the cardinality of
$A$ and $B$ is $| A | = | B | = s$, for an arbitrary large $s \in \N$.  We
slightly abuse notation and denote both the vertices in $A$ and in $B$
by the integers $1,\ldots,s$.  If we need to distinguish $i \in A$
from $i \in B$, we use the notation $i^A$ and $i^B$.  We may 
assume that $E$ consists of the edges that connect $i \in A$ with $i
\in B$ for every $i \in \set{1,\ldots,s}$, i.e., $E =
\Set{\set{i^A,i^B}}{1 \leq i \leq s}$.  These edges in $E$ are the
{\em pipes} in the above terminology.  We now extend the graph to
$(A_\circ \cup B,E)$ by adding a vertex $0$ to $A$, resulting in
$A_\circ = A \cup \set{0}$.  This vertex corresponds to the {\em water
  tap}, which Alice can connect to one of the pipes.  Given a Boolean
function $f: \set{0,1}^n \times \set{0,1}^n \to \set{0,1}$, consider
now two functions $E_{A_\circ}$ and $E_B$; both take as input a string
in $\set{0,1}^n$ and output a set of edges (without self loops). For
any $x,y \in \set{0,1}^n$, $E_{A_\circ}(x)$ is a set of edges on the
vertices $A_\circ$ and $E_B(y)$ is a set of edges on the vertices $B$,
so that the resulting graphs $(A_\circ,E_{A_\circ}(x))$ and
$(B,E_B(y))$ have maximum degree at most $1$.  $E_{A_\circ}(x)$
consists of the {\em connections} among the pipes (and the tap) on
Alice's side (on input $x$), and correspondingly for $E_B(y)$.  For
any $x,y \in \set{0,1}^n$, we now define the graph $G(x,y) = (A_\circ
\cup B,E \cup E_{A_\circ}(x) \cup E_B(y))$ by adding the edges
$E_{A_\circ}(x)$ and $E_B(y)$ to $E$.  $G(x,y)$ consists of the pipes
with the connections added by Alice and Bob.
Note that the vertex $0 \in A_\circ$ has degree at most $1$, and the
graph $G(x,y)$ has maximum degree at most two $2$; it follows that the
maximal path $\pi(x,y)$ that starts at the vertex $0 \in A_\circ$ is
uniquely determined. $\pi(x,y)$ represents the flow of the water, and
the endpoint of $\pi(x,y)$ determines whether the water comes out on Alice
or on Bob's side (depending on whether it is in $A_\circ$ or in $B$).

\begin{definition}
  A {\bf garden-hose game} is given by a graph function $G: (x,y)
  \mapsto G(x,y)$ as described above.  The number of pipes $s$ is
  called the {\bf size} of $G$, and is denoted as $s(G)$.  A
  garden-hose game $G$ is said to {\bf compute} a Boolean function $f:
  \set{0,1}^n \times \set{0,1}^n \to \set{0,1}$ if the endpoint of the
  maximal path $\pi(x,y)$ starting at $0$ is in $A_\circ$ whenever
  $f(x,y) = 0$ and in $B$ whenever $f(x,y) = 1$.
\end{definition}


\begin{definition}
The (deterministic) {\bf garden-hose complexity} of a Boolean function
$f: \set{0,1}^n \times \set{0,1}^n \to \set{0,1}$ is the size $s(G)$
of the smallest garden-hose game $G$ that computes $f$. We denote it
by $\gh(f)$. 
\end{definition}

We start with a simple upper bound on $\gh(f)$ which is implicitly
proven in the attack on Scheme II in~\cite{Kent2010}].
\begin{proposition}
  For every Boolean function $f: \set{0,1}^n \times \set{0,1}^n \to
  \set{0,1}$, the garden-hose complexity is at most $\gh(f) \leq 2^{n}+1$.
\end{proposition}
\begin{proof}
  We identify $\set{0,1}^n$ with $\set{1,\ldots,2^n}$ in the natural
  way. For $s = 2^n+1$ and the resulting bipartite graph $(A_\circ
  \cup B,E)$, we can define $E_{A_\circ}$ and $E_B$ as
  follows. $E_{A_\circ}(x)$ is set to $\set{(0,x)}$, meaning that Alice connects the tap with the pipe labeled by her input $x$. To define $E_B$,
  group the set $Z(y) = \Set{a \in \set{0,1}^n}{f(a,y)=0}$ arbitrarily
  into disjoint pairs $\set{a_1,a_2} \cup \set{a_3,a_4} \cup \ldots
  \cup \set{a_{\ell-1},a_\ell}$ and set $E_B(y) =
  \set{\set{a_1,a_2},\set{a_3,a_4}, \ldots, \set{a_{\ell-1},a_\ell}}$.
  If $\ell = |Z(y)|$ is odd so that the decomposition into pairs
  results in a left-over $\set{a_\ell}$, then $a_\ell$ is connected
  with the ``reserve'' pipe labeled by $2^n+1$.

By construction, if $x \in Z(y)$ then $x = a_i$ for some $i$, and thus
pipe $x = a_i$ is connected on Bob's side with pipe $a_{i-1}$ or
$a_{i+1}$, depending on the parity of $i$, or with the ``reserve''
pipe, and thus $\pi(x,y)$ is of the form $\pi(x,y) =
(0,x^A,x^B,v^B,v^A)$, ending in $A_\circ$.  On the other hand, if $x
\not\in Z(y)$, then pipe $x$ is not connected on Bob's side, and thus $\pi(x,y) =
(0,x^A,x^B)$, ending in $B$. This proves the claim.  
\end{proof}
We notice that the same proof shows that the garden-hose complexity
$\gh(f)$ is at most $2^k + 1$, when $k$ is the one-way communication complexity from Alice
to Bob of $f$.\footnote{Or if needed, with a small adjustment in the protocol,
$2^k+2$ with $k$ the one-way communication complexity of Bob to Alice.}

We introduce the following terminology. We say that a function $f: \set{0,1}^n \times \set{0,1}^n \to \set{0,1}$ is {\em obtained} from a function $g: \set{0,1}^m \times \set{0,1}^m \to \set{0,1}$ {\em by local pre-processing} if $f$ is of the form $f(x,y) = g(\alpha(x),\beta(y))$, where $\alpha$ and $\beta$ are arbitrary functions $\set{0,1}^n \to \set{0,1}^m$. 
The following invariance under local preprocessing follows immediately from the definition of the garden-hose complexity. 

\begin{lemma}\label{lemma:prepro}
If $f$ is obtained from $g$ by local pre-processing, then $\gh(f) \leq \gh(g)$. 
\end{lemma}

\section{Garden-Hose Complexity of Specific Functions}\label{sec:specificfunctions}

To get a feel for the kind of things that are possible in the garden-hose model,
we will first look at upper bounds for the complexity of several functions that are often
studied in communication complexity settings:
\begin{itemize}
\item Equality: $\mathrm{EQ}(x,y)= 1$ iff $x=y$
\item Bitwise inner product: $\mathrm{IP}(x,y)=\sum_i x_i y_i \pmod{2}$
\item Majority function: $\mathrm{MAJ}(x,y)=1$ iff $\sum_i x_i y_i \geq \lfloor \frac{n}{2} \rfloor+1$
\end{itemize}

\subsection{Equality}
For a graphical depiction of the protocol, see Figure~\ref{fig:eq}.
As initialization, Alice first connects the source to pipe $R_0$, effectively
letting Bob start with the water.

We repeat the same pattern, for every $i$ from $1$ to $n$.
If $y=0$, Bob connects pipe $R_{i-1}$ to pipe $Q^0_i$, and on 
$y=1$, Bob connects pipe $R_{i-1}$ to pipe $Q^1_i$.
On the other side, Alice connects $R_i$ to $Q^0_i$ if $x=0$ and she
connects $R_i$ to $Q^1_i$ instead, if $x=1$.

If $x$ and $y$ are different on bit $j$, then $Q^{y_j}_j$ stays unconnected,
so the water will flow out on Alice's side, right there. If $x$ and $y$ are equal
this situation will never happen, so the water will exit at $R_n$, on Bob's side.
The strategy uses $3n+1$ pipes, so we have shown that
\[
\gh(\mathrm{EQ}) \leq 3n+1 \text{.}
\]

\begin{figure}[h!tb]
\begin{framed}
\center
\includegraphics{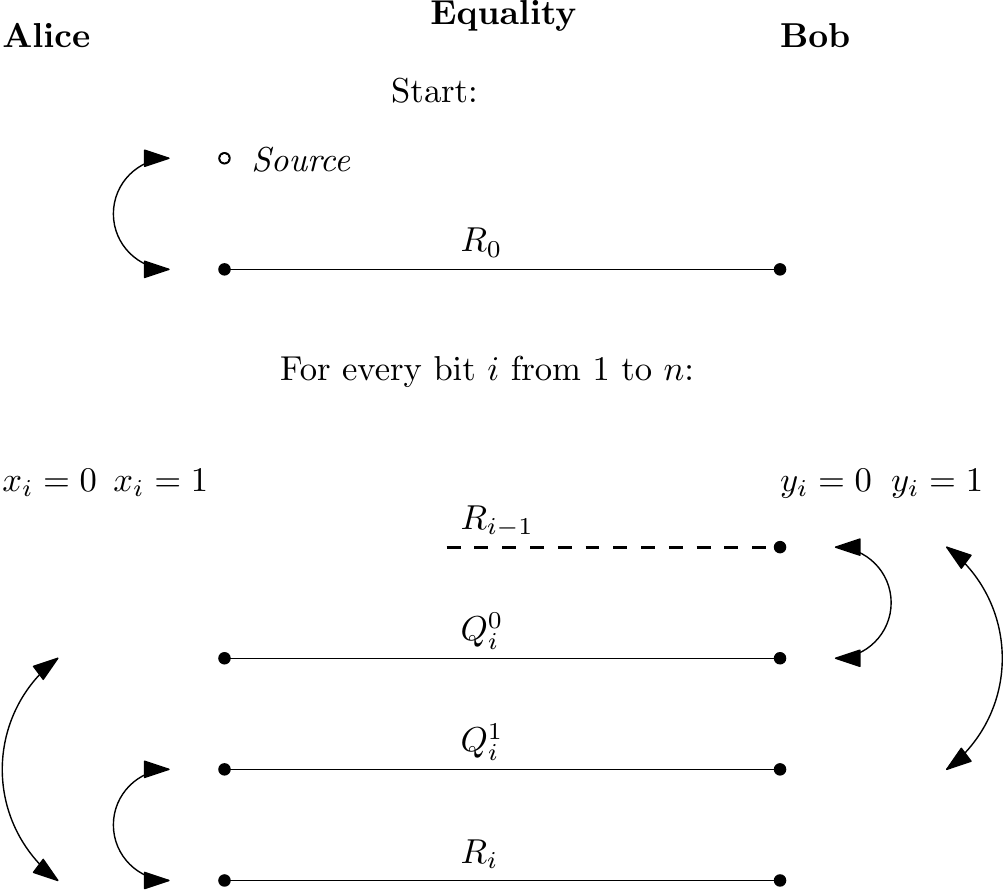}
\end{framed}
\caption{Garden-hose game for the equality function. \label{fig:eq}}
\end{figure}

\subsection{Inner product}
The protocol for inner product is drawn in Figure~\ref{fig:ip}.
Recall that the inner-product function is defined as $\mathrm{IP}(x,y)=\sum_i x_i y_i \pmod{2}$.
To calculate the bitwise inner product, we might let $i$ go from $1$ to $n$, initialize a one-bit result
register with the value $0$, and flip this bit whenever the AND of $x_i$ and $y_i$ equals $1$.
The garden-hose protocol follows a strategy inspired by this idea.

To start, Alice connects the source to $Q^0_k$, with $k$ the first index for which $x_k=1$.
For every $i$ from $1$ to $n$, there are four pipes. If $y_i=0$, Bob connects $Q^0_i$ to $R^0_i$ and
$Q^1_i$ to $R^1_i$. If $y=1$, Bob instead connects $Q^0_i$ to $R^1_i$ and
$Q^1_i$ to $R^0_i$.

Alice does not make any new connections if $x_i=0$, and if $x_i=1$ she connects $R^0_i$ and $R^1_i$ to
$R^0_k$ and $R^1_k$ respectively, with $k$ the next index for which $x_k=1$. If $x_i$ is the
last bit of $x$ equal to 1, Alice does nothing with $R^0_i$ and connects $R^1_i$ to the pipe
labeled $\mathit{End}$.

To see why this construction works, we can compare it to the algorithm
described earlier.
The water flowing through $R^b_i$ corresponds to the result register having value $b$
after step $i$ of the algorithm, and the water changes from the top to bottom
pipe, or vice versa, when $x_i = y_i = 1$. At the last index $k$ for which $x_k=1$,
the water flows to Alice through the pipe corresponding to the final function value.
Alice leaves $R^0_k$ unconnected, so the water exits at Alice's side if $\mathrm{IP}(x,y)=0$.
She connects $R^1_k$ to the pipe $\mathit{End}$, which is unconnected on Bob's side,
making the water exit at Bob's side if $\mathrm{IP}(x,y)=1$.

The strategy uses $4n+1$ pipes, letting
us upper bound the garden-hose complexity with
\[
\gh(\mathrm{IP}) \leq 4n+1 \text{.}
\]

\begin{figure}[h!tb]
\begin{framed}
\center
\includegraphics{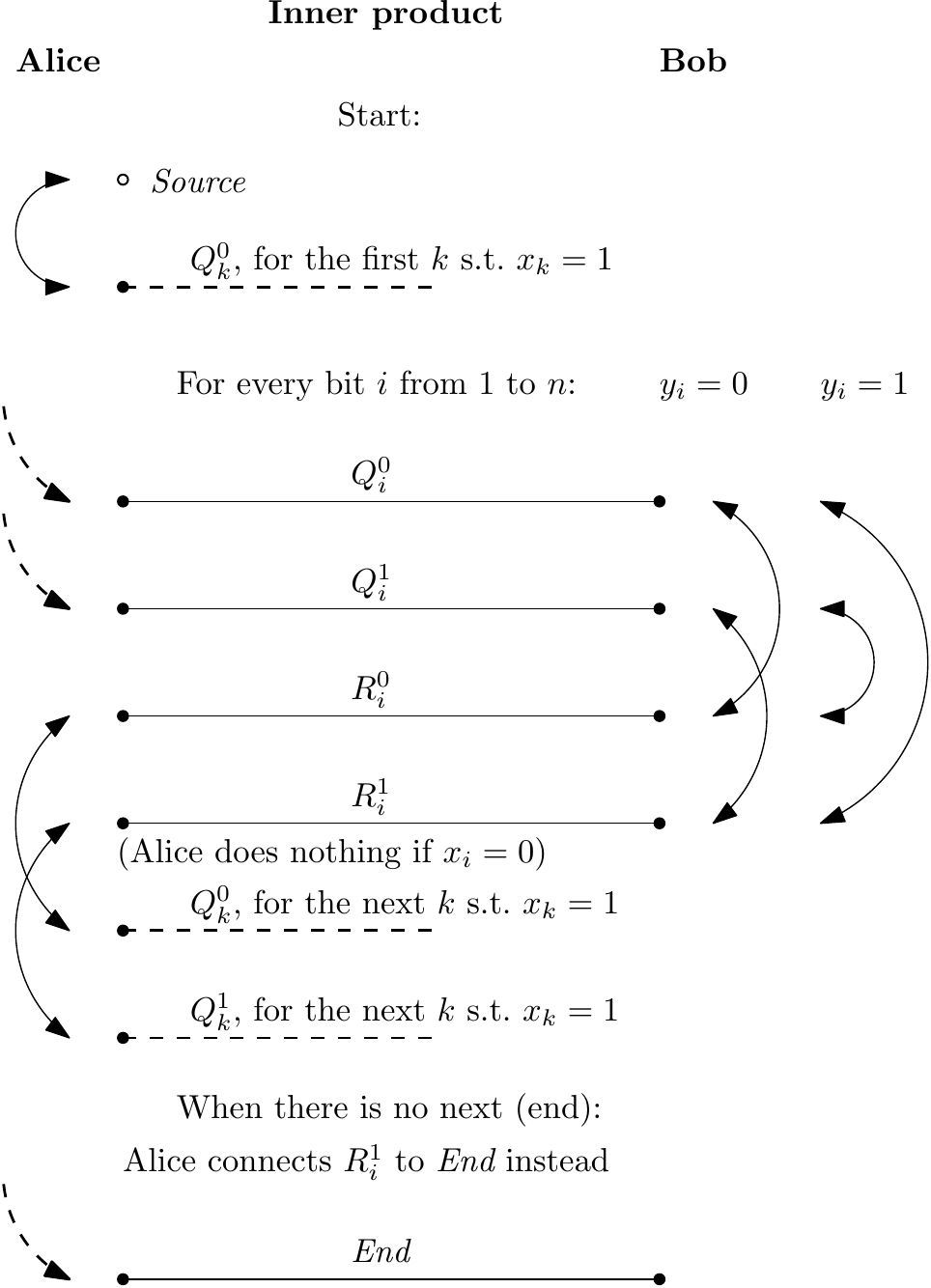}
\end{framed}
\caption{Garden-hose game for the inner-product function. \label{fig:ip}}
\end{figure}

\subsection{Majority}\label{sec:majority}
The majority function equals 1 when on at least half
of indices $i$ we have that both $x_i=1$ and $y_i=1$.
Our strategy in the garden-hose game is inspired by the following algorithm
for majority. We iterate over all indices $i$ from $1$ to $n$ and
initialize a counter with value 0.
For every $i$, we add $1$ to the counter if both $x_i=1$ and $y_i=1$.
If the value in the counter reaches $\left\lfloor \frac{n}{2}\right\rfloor+1$, we stop and answer 1.
Here $\left\lfloor z \right\rfloor$ is the floor function which maps
a real number $z$ to the largest integer not greater $z$. Otherwise, if we reach the end, we give the answer 0.
Our garden-hose strategy works in a similar way, with the pipe the water flows through
acting as a `counter'. For simplicity, we assume that $n$ is a multiple of 2. It is easy to extend
the strategy to also work for odd $n$.

See Figure~\ref{fig:maj} for a diagram illustrating the strategy Alice and Bob follow.
Alice connects the source to $Q^0_k$, with $k$ the first index for which $x_k=1$.
For every $i$ from $1$ to $n$, the players have $n/2+1$ pipes labeled $Q^0_i, \dots Q^{n/2}_i$ and
$n/2+1$ pipes labeled $R^0_i, \dots R^{n/2}_i$. If $y_i=0$, Bob connects $Q^m_i$ to $R^m_i$ for every
$m$ from $0$ to $n/2$. If $y_i=1$, Bob instead connects $Q^m_i$ to $R^{m+1}_i$ for every
$m$ from $0$ to $n/2-1$, and leaves $Q^{n/2}_i$ unconnected.

If $x_i=0$, Alice does not make any new connections. Now, look at the case where $x_i=1$.
Let $k$ be the first index greater than $i$ for
which $x_k=1$. If $i$ is the last index for which $x_i=1$, Alice does nothing.
Otherwise, she connects $R^m_i$ to $Q^m_k$, for every $m$ from $0$ to $n/2$.

Having the earlier algorithm in mind, we can see how the garden-hose
strategy works by comparing it to the algorithm.
The water flowing through pipe $Q^c_i$ corresponds to
the counter having value $c$ at step $i$. Where both $x_i=1$ and $y_i=1$, the water will go
to Bob's side in pipe $Q^c_i$ and return to Alice in pipe $R^{c+1}_i$. The water going back in a 
lower pipe is equivalent to incrementing the counter. When the water was coming in through $Q^{n/2}_i$ and $y_i=1$,
the water will exit at Bob's side, since $Q^{n/2}_i$ will be unconnected then.
The water exiting at Bob
corresponds to stopping and answering 1 when the counter reaches $n/2+1$. Finally, if there are not enough
positions $i$ where both $x_i=1$ and $y_i=1$,
the water will exit at Alice at the last $i$ for which $x_i=1$. In the algorithm this case
is equivalent to outputting 0 if the end is reached without terminating earlier.

This strategy uses $n+2$ pipes for every $i$, giving a total upper bound of
\[
\gh(\mathrm{MAJ}) \leq (n+2)^2 \text{.}
\]
It is not hard to get a strategy with approximately half this number, we will give a sketch on how to 
modify the strategy to achieve this improvement.
For $i$ values $1$ to $n/2-1$ we only need pipes $Q^{0}_i$ to $Q^{i}_i$ and $R^{0}_i$ to $R^{i+1}_i$.
We can do leave out the other pipes because the counter can not have reached the corresponding value,
even if all bits of $x$ and $y$ so far were 1.
When $i$ has values $n/2$ to $n$ we only need $Q^{i-n/2}_i$ to $Q^{n/2-1}_i$ and $R^{i+1-n/2}_i$ to $R^{n/2-1}_i$.
Leaving these pipes out is possible because for low values,
at that step, the counter will not be able to reach $n/2+1$, even if all remaining bits of $x$ and $y$ are 1.
We did not include this improvement in the main strategy to keep it simpler, while still keeping
the same upper bound of $O(n^2)$, up to constants.

\begin{figure}[h!tb]
\begin{framed}
\center
\includegraphics{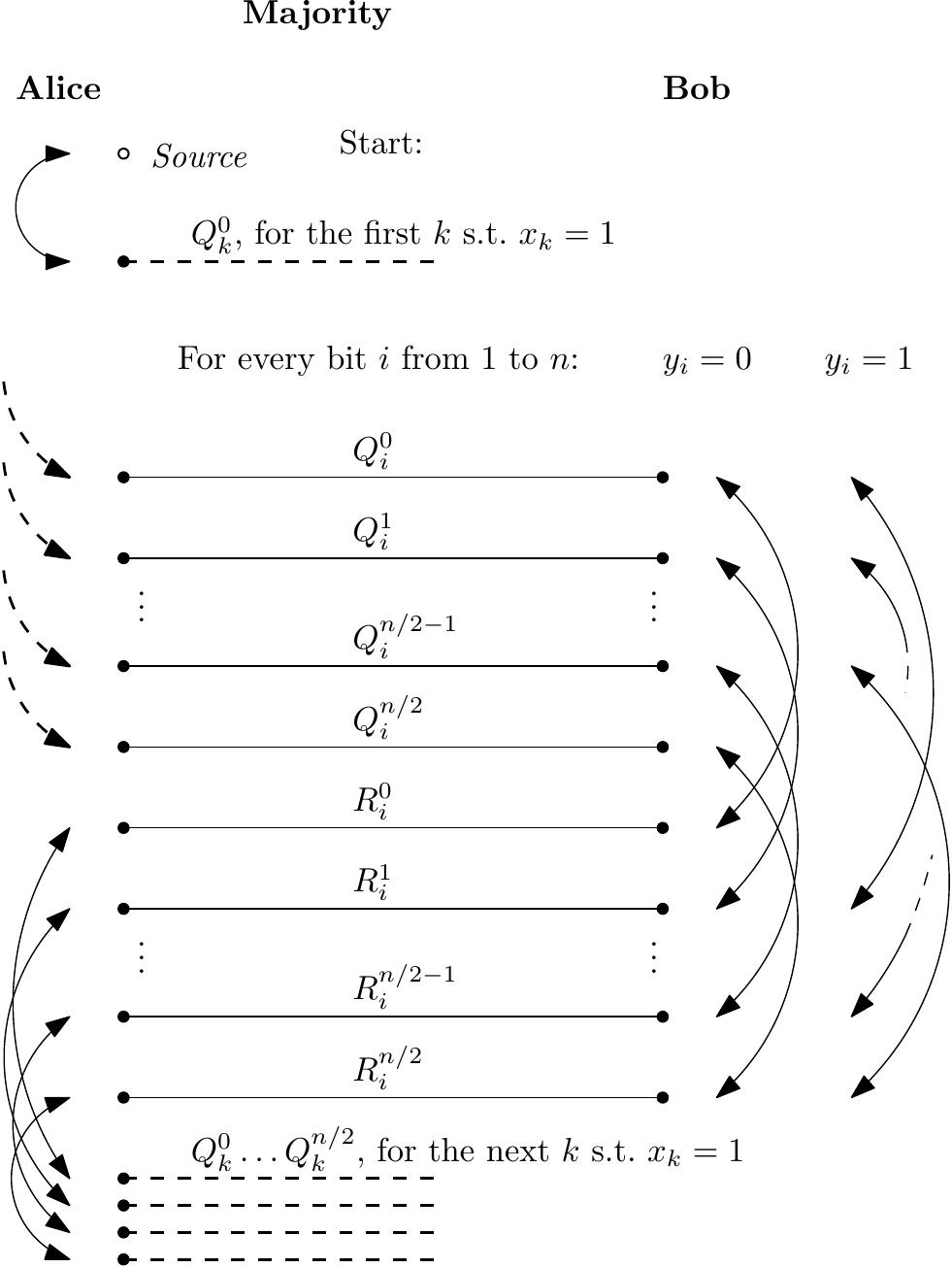}
\end{framed}
\caption{Garden-hose game for the majority function. Note that the $Q^{n/2}_i$ pipe
is left unconnected on Bob's side whenever $y_i=1$.\label{fig:maj}}
\end{figure}

An interesting thing to note is that equality and inner product
can be computed in $n$ steps using constant memory, and we are able to find a
garden-hose strategy using $O(n)$ steps. The obvious way to compute the majority function
needs one counter, using $\log{n}$ memory, and the given upper bound for the garden-hose complexity
is $O(n^2)$. This already hints at the result given in Section~\ref{sec:logspace}, where we
show that any function that can be computed in logarithmic space has at most polynomial
garden-hose complexity.

\section{Encoding Logarithmic-Depth Circuits}\label{sec:logdepthcircuits}
For any function that can be computed by a circuit that has
depth logarithmic in the input length, we can find a strategy.
We use a construction inspired by Barrington's theorem~\cite{Barrington1989},
for which a proof is given in Section~\ref{sec:barringtonstheorem}.
Even though the notation coming from the machinery of Barrington's theorem is a bit involved,
the actual construction in the garden-hose model matches the \emph{permutation branching programs},
whose power follows from Barrington's theorem.

\begin{theorem}
If $f: \set{0,1}^n \times \set{0,1}^n \to \set{0,1}$ is computable in $\NCone$,
then $\gh(f)$ is polynomial in $n$.
\end{theorem}
\begin{proof}
Define $z \in \set{0,1}^{2n}$ as the concatenation of the inputs of $f(x,y)$, where
$x$ is the part of the input Alice holds and Bob has $y$. We
say that $f$ is computable in $\NCone$ if there exists a circuit $C$ with the following properties.
The depth $d$ of $C$ is logarithmic in the input length $2n$. The circuit $C$ has
fan-in 2, and consists of only NOT, AND and OR gates. $C$ outputs 1 if and only if $f$ outputs 1 on the
same input.
 
We define an \emph{instruction} as a triple $(j, \mu, \nu)$, where $j$ is
the index to a bit of the input, and $\mu$ and $\nu$ are permutations
in the symmetric group $S_5$. We define a \emph{program} as a list of instructions. A
program evaluates to the product of the value of its instructions. We say that a program 
\emph{computes} a circuit $C$ if it evaluates to a 5-cycle $\mu$ when $C$ is true,
and to the identity $e$ when $C$ is false.

It follows from Barrington's theorem that, given $C$, we can construct a program
$P$ with length $l$ at most $4^d$. This $P$ computes $f$, and has
length polynomial in $n$. Without loss of generality we can assume that the instructions alternate
between depending on $x$ and depending on $y$. If the instructions do not alternate, it will be easy
to modify the construction so that the players can locally collect multiple instructions together.
We also assume that the number of instructions $l$ is even.

Let $P_i$ be the $i$-th instruction of $P$. Alice can evaluate all the odd instructions and Bob knows all the even
instructions. Recall that if $f(x,y)=0$, the product of these evaluations is the identity permutation $e$,
and on $f(x,y)=1$, the product is some other 5-cycle. Let $P_i(a)$ be the evaluation of the $i$-th instruction
acting on the number $a$. Label pipes $Q^1_1$, $Q^1_2$, $Q^1_3$, $Q^1_4$, $Q^1_5$ up to
$Q^l_1$, $Q^l_2$, $Q^l_3$, $Q^l_4$, $Q^l_5$, with $l$ the length of $P$.

First, Alice evaluates $P_1$ and connects the source to pipe $Q^1_{P_1(1)}$.
Then, for every odd $i$ up to $l$, Alice connects pipe $Q^i_k$ to pipe
$Q^{i+1}_{P_i(k)}$, for $k$ from $1$ to $5$; she connects the pipes according the permutation
given by the instructions. Because all the odd instructions depend on $x$, she is able to
find $P_i$ for every odd $i$.
Bob's actions are similar: for every even $i$ up to $l$,
and $k$ from $1$ to $5$, Bob connects pipe $Q^i_k$ to pipe
$Q^{i+1}_{P_i(k)}$.
At the end, Alice leaves $Q^l_1$ unconnected and uses 4 pipes to let $Q^l_2, \dots, Q^l_5$ go to Bob.

Because we linked up the groups of 5 pipes according to the permutations given by the permutation branching program,
if $f(x,y)=0$ the identity permutation will be applied in total, so water will flow through $Q^l_1$, correctly
exiting at Alice's side.
Otherwise, if $f(x,y)=1$, the water will go through one of the other pipes, since a cycle permutation does not
leave any position unchanged, correctly letting the water flow to Bob.
\end{proof}

\section{Garden-Hose Complexity and Log-Space Computations}
\label{sec:logspace}

The following theorem shows that for a large class of functions, a
polynomial amount of pipes suffices to win the garden-hose game. A function $f$
with an $n$-bit input
is log-space computable if there is a deterministic Turing machine $M$ and a constant $c$, such that 
$M$ outputs the correct value of $f$, and at most $c \cdot \log{n}$ locations
of $M$'s work tapes are ever visited by $M$'s head during computation of every input of length $n$.

\begin{theorem} \label{thm:logspace}
  If $f: \set{0,1}^n \times \set{0,1}^n \to \set{0,1}$ is log-space computable, 
then $\gh(f)$ is polynomial in $n$.
\end{theorem}

In combination with Lemma~\ref{lemma:prepro}, it follows immediately that the same conclusion also holds for functions that are {\em log-space computable up to local pre-processing}, i.e., for any function $f: \set{0,1}^n \times \set{0,1}^n \to \set{0,1}$ that is obtained from a log-space computable function $g: \set{0,1}^m \times \set{0,1}^m \to \set{0,1}$ by local pre-processing, where $m$ is polynomial in $n$. 
Below, in Proposition~\ref{prop:logspaceinv}, we show that log-space up to local pre-processing is also {\em necessary} for a polynomial garden-hose complexity.  

We will later see (Proposition~\ref{prop:expbound}) that there exist functions with large garden-hose complexity.
However, a negative implication of Theorem~\ref{thm:logspace} is that proving the
existence of a {\em polynomial-time computable} function $f$ with exponential garden-hose complexity is at least as hard
as separating $\class{L}$ from $\class{P}$, a long-standing open problem in complexity theory.

\begin{corollary}
If there exists a function $f: \set{0,1}^n \times \set{0,1}^n \to \set{0,1}$ in P that has superpolynomial garden-hose complexity, then {\rm P} $\neq$ {\class{L}}.  
\end{corollary}

\begin{proof}[Proof of Theorem~\ref{thm:logspace}]
Let $M$ be a deterministic Turing machine deciding $f(x,y)=0$. We
assume that $M$'s read-only input tape is of length $2n$ and contains
$x$ on positions $1$ to $n$ and $y$ on positions $n+1$ to $2n$. By
assumption $M$ uses logarithmic space on its work tapes.

In this proof, a \emph{configuration} of $M$ is the location of its
tape heads, the state of the Turing machine and the content of its
work tapes, excluding the content of the read-only input tape.  This
is a slightly different definition than usual, where the content of
the input tape is also part of a configuration. When using the normal
definition (which includes the content of all tapes), we will use the term
\emph{total configuration}. Any configuration of $M$ can be described
using a logarithmic number of bits, because $M$ uses logarithmic
space.

A Turing machine is called \emph{deterministic}, if every total
configuration has a unique next one. A Turing machine is called
\emph{reversible} if in addition to being deterministic, every total configuration also has a unique predecessor.
An $S(n)$ space-bounded deterministic
Turing machine can be simulated by a reversible Turing machine in
space $O(S(n))$~\cite{Lange1998}.
This means that without
loss of generality, we can assume $M$ to be a reversible Turing
machine, which is crucial for our construction. Let $M$ also be
\emph{oblivious}\footnote{A Turing machine is called \emph{oblivious},
  if the movement in time of the heads only depend on the length of the input, known in advance to be $2n$, but
  not on the input itself. For our construction we only require the input tape head to have this property.} in the tape head movement on the input tape. This can be done with only a small increase in space by adding a
counter.

Alice's and Bob's perfect strategies in the garden-hose game are as
follows. They list all configurations where the head of the input tape
is on position $n$ coming from position $n+1$. Let us call the set of
these configurations $C_{A}$. Let $C_{B}$ be the analogous set of
configurations where the input tape head is on position $n+1$ after
having been on position $n$ the previous step. Because $M$ is
oblivious on its input tape, these sets depend only on the function
$f$, but not on the input pair $(x,y)$.  The number of elements of $C_A$
and $C_B$ is at most polynomial, being exponential in the description
length of the configurations.  Now, for every element in $C_{A}$ and
$C_{B}$, the players label a pipe with this configuration.  Also label $|C_{A}|$
pipes $\accept$ and $|C_{B}|$ of them $\reject$.  These steps determine
the number of pipes needed, Alice and Bob can do this labeling
beforehand.

For every configuration in $C_{A}$, with corresponding pipe $p$, Alice
runs the Turing machine starting from that configuration until it
either accepts, rejects, or until the input tape head reaches position
$n+1$. If the Turing machine accepts, Alice connects $p$ to the first free pipe labeled $\accept$.
On a reject, she leaves $p$ unconnected. If
the tape head of the input tape reaches position $n+1$, she connects
$p$ to the pipe from $C_{B}$ corresponding to the configuration of the
Turing machine when that happens. By her knowledge of $x$, Alice knows
the content of the input tape on positions $1$ to $n$, but not the
other half.
Alice also runs $M$ from the starting configuration, connecting the
water source to a target pipe with a configuration from $C_{B}$
depending on the reached configuration.

Bob connects the pipes labeled
by $C_{B}$ in an analogous way: He runs the Turing machine starting
with the configuration with which the pipe is labeled until it halts
or the position of the input tape head reaches $n$. On accepting, the
pipe is left unconnected and if the
Turing machine rejects, the pipe is connected to one of the pipes labeled $\reject$. Otherwise, the
pipe is connected to the one labeled with the configuration in
$C_{A}$, the configuration the Turing machine is in when the head on
the input tape reached position $n$.

In the garden-hose game, only one-to-one connections of pipes are
allowed. Therefore, to check that the described strategy is a valid one,
the simulations of two different configurations
from $C_A$ should never reach the same configuration in
$C_B$. This is guaranteed by the reversibility of $M$ as follows. 
Consider Alice simulating $M$ starting from different configurations $c
\in C_A$ and $c' \in C_A$. We have to check that their simulation can not end at
the same $d \in C_B$, because Alice can not connect both pipes labeled $c$ and $c'$ to the same
$d$. Because $M$ is reversible, we can in principle also simulate $M$ backwards in time starting
from a certain configuration. In particular, Alice can simulate $M$ backwards starting with configuration $d$,
until the input tape head position reaches $n+1$. The configuration of $M$ at that time can not simultaneously
be $c$ and $c'$, so there will never be two different pipes trying to connect to the pipe labeled $d$.

It remains to show that, after the players link up their pipes as described,
the water comes out on Alice's side if $M$ rejects on input $(x,y)$,
and that otherwise the water exits at Bob's.  We can verify the correctness of the described strategy
by comparing the flow of the water directly to the execution of $M$. Every pipe the
water flows through corresponds to a configuration of $M$ when it runs starting from the
initial state. So the side on which the water finally exits also corresponds to whether $M$ accepts or rejects.
\end{proof}

\begin{proposition} \label{prop:logspaceinv}
Let $f: \set{0,1}^n \times \set{0,1}^n \to \set{0,1}$ be a Boolean function. If $\gh(f)$ is polynomial (in $n$), then $f$ is log-space computable up to local pre-processing. 
\end{proposition}

\begin{proof}
Let $G$ be the garden-hose game that achieves $s(G) = \gh(f)$. We write $s$ for $s(G)$, the number of pipes, and we let $E_{A_\circ}$ and $E_B$ be the underlying edge-picking functions, which on input $x$ and $y$, respectively, output the connections that Alice and Bob apply to the pipes. 
Note that by assumption, $s$ is polynomial. Furthermore, by the restrictions on $E_{A_\circ}$ and $E_B$, on any input, they consist of at most $(s+1)/2$ connections. 

We need to show that $f$ is of the form $f(x,y) = g(\alpha(x),\beta(y))$, where $\alpha$ and $\beta$ are arbitrary functions $\set{0,1}^n \to \set{0,1}^m$, $f: \set{0,1}^m \times \set{0,1}^m \to \set{0,1}$ is log-space computable, and $m$ in polynomial in $n$. We define $\alpha$ and $\beta$ as follows. For any $x,y \in \set{0,1}^n$, $\alpha(x)$ is simply a natural encoding of $E_{A_\circ}(x)$ into $\set{0,1}^m$, and $\beta(y)$ is a natural encoding of $E_B(y)$ into $\set{0,1}^m$.
In the hose-terminology we say that $\alpha(x)$ is a binary encoding of the connections of Alice,
and $\beta(y)$ is an encoding of the connections of Bob.
Obviously, this can be done with $m$ of polynomial size. Given these encodings, finding the endpoint of the maximum path $\pi(x,y)$ starting in $0$ can be done with logarithmic space: at any point during the computation, the Turing machine only needs to maintain a couple of pointers to the inputs and a constant number of binary flags. 
Thus, the function $g$ that computes $g(\alpha(x),\beta(y)) = f(x,y)$ is log-space computable in $m$ and thus also in $n$. 
\end{proof}

\section{Lower Bounds}\label{sec:ghlowerbounds}
In this section, we present lower bounds on the number of pipes
required to win the garden-hose game for particular (classes of) functions.

\begin{definition}\label{def:injective} We call a function $f$ \emph{injective for Alice}, if for every two different inputs $x$ and $x'$ there exists $y$ such
  that $f(x,y) \neq f(x',y)$. We define \emph{injective for Bob} in an
  analogous way: for every $y \neq y'$, there exists $x$ such that
  $f(x,y) \neq f(x,y')$ holds.
\end{definition}

\begin{proposition}
  If $f$ is injective for Bob or $f$ is injective for Alice,
  then
  $$\gh(f) \log(\gh(f)) \geq n \, .$$
\end{proposition}
\begin{proof}
  We give the proof when $f$ is injective for Bob. The proof for the
  case where $f$ is injective for Alice is the same.  Consider a
  garden-hose game $G$ that computes $f$. Let $s$ be its size
  $s(G)$. Since, on Bob's side, every pipe is connected to at most one other pipe, there are at most $s^s = 2^{s \log(s)}$
  possible choices for $E_B(y)$, i.e., the set of connections on Bob's side. Thus, if $2^{s \log(s)} < 2^n$, it
  follows from the pigeonhole principle that there must exist $y$ and
  $y'$ in $\set{0,1}^n$ for which $E_B(y) = E_B(y')$, and thus for
  which $G(x,y) = G(x,y')$ for all $x \in \set{0,1}^n$. But this
  cannot be since $G$ computes $f$ and $f(x,y)\neq f(x,y')$ for some
  $x$ due to the injectivity for Bob. Thus, $2^{s \log(s)} \geq 2^n$
  which implies the claim. 
\end{proof}

We can use this result to obtain an almost linear lower bound for the functions we looked at in
Section~\ref{sec:specificfunctions}. The bitwise inner product, equality and majority functions
are all injective for both Alice and Bob, giving us the following corollary.

\begin{corollary}
The functions bitwise inner product, equality and majority have garden-hose complexity at least
$\frac{n}{\log (n)}$.
\end{corollary}

\begin{proposition}\label{prop:expbound}
  There exist functions $f:\set{0,1}^n \times \set{0,1}^n \to
  \set{0,1}$ for which $\gh(f)$ is exponential.
\end{proposition}

\begin{proof}
  The existence of functions with an exponential garden-hose complexity
  can be shown by a simple counting argument. There are
  $2^{2^{2n}}$ different functions $f(x,y)$. For a given size $s =
  s(G)$ of $G$, for every $x \in \set{0,1}^n$, there are at most
  $(s+1)^{s+1}$ ways to choose the connections $E_{A_\circ}(x)$ on Alice's
  side, and thus there are at most $((s+1)^{s+1})^{2^n} = 2^{2^n (s+1)
    \log(s+1)}$ ways to choose the function $E_{A_\circ}$. Similarly
  for $E_B$, there are at most $2^{2^n s \log(s)}$ ways to choose
  $E_B$. Thus, there are at most $2^{2\cdot 2^n (s+1) \log(s+1)}$ ways
  to choose $G$ of size $s$. Clearly, in order for every function $f$
  to have a $G$ of size $s$ that computes it, we need that $2\cdot 2^n
  (s+1) \log(s+1) \geq 2^{2n}$, and thus that $(s+1) \log(s+1) \geq
  2^{n-1}$, which means that $s$ must be exponential.
\end{proof}


\section{Notes on feasibility}

It is very common in complexity theory to say that an algorithm is
efficient when it uses only a polynomial amount of resources. This
is also the spirit of the upper bound given in Section~\ref{sec:logspace},
where we showed that the garden-hose complexity of a function that
can be computed in logarithmic space is at most polynomial. For the
task of secure position verification however, a real-world adversary
is quite limited. A protocol for which the best known attack requires
an number of EPR pairs that is almost linear in $n$ (the number of
classical bits) will certainly be not breakable with current technology.
The actions of the honest players, on the other hand, are within reach
to implement, and the basic steps are not much harder than those used
in, for example, the BB84 protocol~\cite{Bennett1984}. If a function should
actually need a quadratic number of EPR pairs to break,
such as the best upper bound we have so far for Majority in section~\ref{sec:majority},
then the corresponding protocol will not be breakable in the foreseeable future.

Of course, we have not proven that the attacks coming from the garden-hose model are optimal;
it might very well be that for some functions there exist quantum
strategies that need much less entanglement than the garden-hose complexity
of that function. We do not know of any such function right now.

Since the honest prover also has to execute the function, the most interesting functions
to look at from a practical perspective will be computable in polynomial time.
Since this thesis has shown that for any function computable in logarithmic space, the dishonest provers
can break the protocol using a polynomial number of EPR pairs, a good candidate 
will be a function $f$ which is in $\class{P}$ but not known to be in $\class{L}$.

Looking for a suitable function gives rise to the following question.
How can we encode the inputs so that the players can not do smart
local pre-processing, i.e.~solve large parts of the problem locally, without needing much of the other half
of the input? One way we propose is to encode the
input with a one-time pad as follows. Given a (hard) function $g(z)$, we define the function
$f$ that is the objective of the garden-hose game as $$f(x,y)=g(x\oplus y) \; .$$
Given an $n$-bit input to the original problem $z$, we give Alice the
random $n$-bit string $r$ and Bob the $n$-bit string $z \oplus r$. 

The strings that Alice and Bob get are both completely random. This makes it harder for
them to smart pre-processing, since every input is equally likely. Even though there
are counter-examples possible, it may be a good option to try.
A disadvantage of this encoding is that we cannot even prove
anymore that any of these functions have exponential garden-hose complexity. The counting argument in Proposition~\ref{prop:expbound} does not work anymore, since we effectively
halve the input length from $2n$ to $n$.

\section{Lower Bounds In The Real World} \label{sec:lowerbound}
In this section, we show that for a function that is injective for
Alice or injective for Bob (according to
Definition~\ref{def:injective}), the dimension of the entangled state
the adversaries need to share in order to attack protocol $\PVqubit$
perfectly has to be of order at least linear in the classical input
size $n$. In other words, they require at least a logarithmic number
of qubits in order to successfully attack $\PVqubit$.

\subsection{Localized Qubits} \label{sec:localized} Assume we have two
bipartite states $\ket{\psi^0}$ and $\ket{\psi^1}$ with the property
that $\ket{\psi^0}$ allows Alice to locally extract a qubit and
$\ket{\psi^1}$ allows Bob to locally extract the same
qubit. Intuitively, these two states have to be different. 

More formally, we assume that both states consist of five registers
$R,A,\tilde{A},B,\tilde{B}$ where registers $R,A,B$ are one-qubit
registers and $\tilde{A}$ and $\tilde{B}$ are arbitrary. We assume
that there exist local unitary transformations $U_{A \tilde{A}}$ and $V_{B
  \tilde{B}}$ such that\footnote{We always assume that these
  transformations act as the identities on the registers we did not
  specify explicitly.}
\begin{align}
U_{A \tilde{A}} \ket{\psi^0}_{RA\tilde{A}B\tilde{B}} &=
\ket{\beta}_{RA} \otimes
\ket{P}_{\tilde{A}B\tilde{B}} \label{eq:unitaryAlice} \\
V_{B \tilde{B}} \ket{\psi^1}_{RA\tilde{A}B\tilde{B}} &=
\ket{\beta}_{RB} \otimes \ket{Q}_{A\tilde{A}\tilde{B}} \, ,
\label{eq:unitaryBob}
\end{align} 
where $\ket{\beta}_{RA} := (\ket{00}_{RA} + \ket{11})_{RA})/\sqrt{2}$
denotes an EPR pair on registers $RA$ and
$\ket{P}_{\tilde{A}B\tilde{B}}$ and $\ket{Q}_{A\tilde{A}\tilde{B}}$
are arbitrary pure states. 

\begin{lemma} \label{lem:innerproduct}
Let $\ket{\psi^0}, \ket{\psi^1}$ be states that
fulfill~\eqref{eq:unitaryAlice} and~\eqref{eq:unitaryBob}. Then,
\[ \big| \, \braket{\psi^0}{\psi^1} \, \big|^2 \leq 1/4 \, . \]
\end{lemma}
\begin{proof}
  Multiplying both sides of~\eqref{eq:unitaryAlice} with
  $U_{A\tilde{A}}^\dag$ and multiplying~\eqref{eq:unitaryBob} with
  $V_{B\tilde{B}}^\dag$, we can write
\begin{align*}
  \big| \, \braket{\psi^0}{\psi^1} \, \big|^2 &= \big| \,
  \bra{\beta}_{RA} \bra{P}_{\tilde{A}B\tilde{B}} \; U_{A\tilde{A}} \,
  V_{B\tilde{B}}^\dag \; \ket{\beta}_{RB} \ket{Q}_{A\tilde{A}\tilde{B}} \,
  \big|^2 \\
  &= \big| \, \bra{\beta}_{RA} \bra{P'}_{\tilde{A}B\tilde{B}}
  \ket{\beta}_{RB} \ket{Q'}_{A\tilde{A}\tilde{B}} \, \big|^2 \; ,
\end{align*}
where we used that $U_{A\tilde{A}}$ and $V_{B\tilde{B}}$ commute and
defined $\ket{P'}_{\tilde{A}B\tilde{B}} := V_{B\tilde{B}}
\ket{P}_{\tilde{A}B\tilde{B}}$ and $\ket{Q'}_{A\tilde{A}\tilde{B}} := U_{A\tilde{A}}
\ket{Q}_{A\tilde{A}\tilde{B}}$.

Without loss of generality we can write
$$\ket{P'}_{\tilde{A}B\tilde{B}} = a \ket{0}_B \ket{P'_0}_{\tilde{A}\tilde{B}} + b \ket{1}_B
  \ket{P'_1}_{\tilde{A}\tilde{B}} \; \text{.}$$
Using that
\begin{align*}
\bra{\beta}_{AB} \big( a \ket{0}_A \ket{z_1}_C + b \ket{1}_A \ket{z_2}_C \big)
 &= \frac{1}{\sqrt{2}} \big( \bra{00}_{AB} + \bra{11}_{AB} \big)
 \big( a \ket{0}_A \ket{z_1}_C + b \ket{1}_A \ket{z_2}_C \big)\\
 &= \frac{1}{\sqrt{2}} \big( a \bra{0}_B \ket{z_1}_C + b \bra{1}_B \ket{z_2}_C \big)\; ,
\end{align*}
for arbitrary registers $A$, $B$ and $C$, and arbitrary quantum states $\ket{z_1}$ and $\ket{z_2}$, we get
\begin{align*}
\big| \, \braket{\psi^0}{\psi^1} \, \big|^2 &= \big| \, \bra{\beta}_{RA} \bra{P'}_{\tilde{A}B\tilde{B}}
  \ket{\beta}_{RB} \ket{Q'}_{A\tilde{A}\tilde{B}} \, \big|^2 \\ 
  &= \big| \, \bra{\beta}_{RA} \big(a^* \bra{0}_B \bra{P'_0}_{\tilde{A}\tilde{B}} + b^* \bra{1}_B
  \bra{P'_1}_{\tilde{A}\tilde{B}}\big)
  \ket{\beta}_{RB} \ket{Q'}_{A\tilde{A}\tilde{B}} \, \big|^2 \\
  &= \frac{1}{2} \big| \, \bra{\beta}_{RA} \big(a^* \ket{0}_R \bra{P'_0}_{\tilde{A}\tilde{B}} + b^* \ket{1}_R
  \bra{P'_1}_{\tilde{A}\tilde{B}}\big)
  \ket{Q'}_{A\tilde{A}\tilde{B}} \, \big|^2 \\
  &= \frac{1}{4} \big| \, \big(a^* \bra{0}_A \bra{P'_0}_{\tilde{A}\tilde{B}} + b^* \bra{1}_A
  \bra{P'_1}_{\tilde{A}\tilde{B}}\big)
  \ket{Q'}_{A\tilde{A}\tilde{B}} \, \big|^2\\
  & \leq \frac{1}{4} \; .
\end{align*}
In the last step we used that the magnitude of the inner product of two quantum states can never exceed 1.
\end{proof}

\subsection{Squeezing Many Vectors in a Small Space}\label{sec:squeeze}
The standard argument from~\cite[Section 4.5.4]{NC00} shows that the
number of unit vectors of pairwise distance $2\eps$ one can fit
into a $d$-dimensional space is of order $\frac{1}{\eps^d}$.

Note that two vectors with absolute inner product equal to
$1/2=\cos(\phi)$ have an angle of $\phi=\arccos(1/2) \approx 1.047$
between them. A small geometric calculation show that they are at
distance $2 \cos(\phi/2) \approx 1.732$ from each other. Hence, we can
consider an $\eps \approx 0.866$ in the statement above.

It follows that if we are trying to squeeze more than
$\frac{1}{\eps^d}$ in a $d$-dimensional space, there will be two
vectors that are closer than $2\eps$ and hence, their inner product is
larger than $1/2$.

\subsection{The Lower Bound}
\begin{theorem}
  Let $f$ be injective for Bob. Assume that Alice and Bob perform a
  perfect attack on protocol $\PVqubit$ where they communicate only
  classical information. Then, they need to pre-share an entangled
  state whose dimension is at least linear in $n$.
\end{theorem}
\begin{proof}
  Let $\ket{\psi}_{RA\tilde{A}B\tilde{B}}$ be the pure state after
  Alice received the EPR half from the verifier. The one-qubit
  register $R$ holds the verifier's half of the EPR-pair, the
  one-qubit register $A$ contains Alice's other half of the EPR-pair,
  the $q_A$-qubit register $\tilde{A}$ is Alice's part of the
  pre-shared entangled state. The registers $B\tilde{B}$ belong to Bob
  where $B$ holds one qubit and $\tilde{B}$ holds $q_B$ qubits. Hence,
  the overall state is a unit vector in a complex Hilbert space of
  dimension $d := 2^{2+q_A+1+q_B}$.

  In the first step of their attack, Alice performs an arbitrary
  quantum operation depending on her classical input $x$ on her
  registers $A\tilde{A}$ resulting in a classical outcome
  $s \in \mathcal{S}$. Similarly, Bob performs a quantum operation depending on $y$ on
  registers $B\tilde{B}$ resulting in classical outcome $t \in \mathcal{T}$. As we
  restricted the players to classical communication, we can assume
  without loss of generality that their operation is a
  measurement.

  We investigate the set $\mathcal{B}$ of overall states after Bob performed his
  measurement, but \emph{before} Alice acts on the state. These states
  depend on Bob's input $y \in \set{0,1}^n$ and his measurement
  outcome $t\in\mathcal{T}$,
\begin{align*}
\mathcal{B} := \Set{ \ket{\psi^{y,t}}_{RA\tilde{A}B\tilde{B}} }{y \in
  \set{0,1}^n,t \in\mathcal{T}} \, .
\end{align*}
The set $\mathcal{B}$ contains at least $2^n$ unit vectors of
dimension $d$. We assume for a contradiction that the dimension $d$ is
smaller than linear in $n$. By the results of
Section~\ref{sec:squeeze}, we know that for $\eps \approx 0.866$, $2^n
> \left(\frac{1}{\eps}\right)^d$ implies that there are two different unit vectors
in $\mathcal{B}$, say $\ket{\psi^{y,t}}$ and $\ket{\psi^{y',t'}}$,
whose absolute inner product is larger than $1/2$.

We now let Alice act on her registers $A\tilde{A}$ of the state. Note
that for every input $x \in \set{0,1}^n$, performing the same action
(depending on $x \in \set{0,1}^n$) with the same outcome $s \in
\mathcal{S}$ on the two states $\ket{\psi^{y,t}}$ and
$\ket{\psi^{y',t'}}$ does not decrease their absolute inner product. Let us call
the states after Alice's actions 
$\ket{\psi^{x,s,y,t}}_{RA\tilde{A}B\tilde{B}}$ and
$\ket{\psi^{x,s,y',t'}}_{RA\tilde{A}B\tilde{B}}$. We have just shown
that for all $x \in \set{0,1}^n, s \in \mathcal{S}$,
\begin{align}\label{eq:largeip}
\big| \braket{\psi^{x,s,y,t}}{\psi^{x,s,y',t'}} \big| > 1/2 \, .
\end{align}

However, because $f$ is injective for Bob, there exists $x$ such that
$f(x,y) \neq f(x,y')$ and hence, the qubit has to end up at different
places depending on Bob's input. For such an $x$ (and arbitrary $s \in
\mathcal{S}$, Lemma~\ref{lem:innerproduct} requires that the states
are ``different'', namely that the absolute inner product $\big|
\braket{\psi^{x,s,y,t}}{\psi^{x,s,y',t'}} \big|$ needs to be smaller
than $1/2$, contradicting~\eqref{eq:largeip}.

Hence, the dimension $d$ of the overall state needs to be at least
linear in $n$.
\end{proof}

\chapter{Open Questions}
In this thesis, we defined the garden-hose model and gave first results for the analysis of a specific scheme
for quantum position-based cryptography. This scheme only requires the honest prover to work with a single qubit,
while the dishonest provers potentially have to manipulate a large quantum state, making it an appealing
scheme to further examine. The garden-hose model captures the power of attacks that only use teleportation,
giving upper bounds for the general scheme, and lower bounds when restricted to these attacks.

The garden-hose model is a new model of communication complexity, and there are still
open questions in relation to this model. Can we find better upper and lower bounds for
the garden-hose complexity of the studied functions?
Our constructions still leave a polynomial gap between lower and upper bounds for
many functions, such as the majority function described in Section \ref{sec:majority}.
It would also be interesting to find an explicit function for which the garden-hose
complexity is provably large, the counting argument in Proposition~\ref{prop:expbound} only
shows the existence of such functions.

Another relevant extension to our results would be the examination of the randomized case:
If we allow Alice and Bob to give the wrong answer with small probability, what are
the lower and upper bounds we can prove in the garden-hose model? For example,
assuming shared randomness between Alice and Bob, we can use results from communication complexity
to show a large gap between the randomized garden-hose complexity of the equality function, 
and the deterministic garden-hose complexity of equality, which we examined in this thesis.

A possible interesting extra restriction on the garden-hose model
would involve limiting the computational power of Alice and Bob.
For example to polynomial time, or the output of quantum circuits of polynomial size.
By bounding not only the amount of entanglement,
but also the amount of computation with a realistic limit, perhaps stronger security proofs are possible.

We can also see multiple interesting open questions when we include the quantum aspects of the problem.
First we have the relation between the garden-hose complexity and the entanglement actually needed
to break the position-verification scheme. Are there quantum attacks on our protocol
for position verification that need asymptotically less entanglement than the garden-hose complexity?
Here it would also be interesting to look at the randomized case. Can we prove lower
bounds, and better upper bounds, if we allow the dishonest provers to make a small error?
The garden-hose lower bounds and quantum lower bounds,
given in this thesis in Section~\ref{sec:ghlowerbounds} and Section~\ref{sec:lowerbound} respectively,
have an exponential gap between them.
Reducing this gap would give more insight into the relative power of all possible
quantum actions to only teleportation, where the garden-hose game captures the power of
attack strategies that just use teleportation.

As a final question, we can ask: How does the protocol behave under parallel repetition?
When executing the protocol once, the dishonest provers 
always have a large probability of cheating the verifiers; even the na\"ive method of measuring the qubit
and distributing the result will work with a probability of at least $0.75$.
By using the protocol multiple times in parallel, given a situation
where the adversaries have a small error,
it might be possible to increase the probability that the dishonest provers are caught 
to arbitrarily close to 1. However, from complexity theory
we know similar situations where provers can achieve
a lower error probability than expected on first sight. In our setting, it remains to be proven
that we can always amplify the probability of the cheaters getting caught.

\appendix

\bibliographystyle{alpha}
\phantomsection\addcontentsline{toc}{chapter}{\bibname}
\bibliography{library}

\end{document}